\documentclass[%
 reprint,
 amsmath,amssymb,
 aps,
]{revtex4-2}

\usepackage{graphicx}% Include figure files
\usepackage{dcolumn}% Align table columns on decimal point
\usepackage{bm}% bold math
\usepackage{fixltx2e}
\usepackage{xcolor}

\usepackage{amsthm}
\usepackage{braket}
\usepackage[english]{babel}
\newtheorem{theorem}{Theorem}
\newtheorem{observation}[theorem]{Observation}
\newtheorem{proposition}[theorem]{Proposition}
\newtheorem{lemma}[theorem]{Lemma}
\newtheorem{corollary}[theorem]{Corollary}
\newcommand{\vtwo}[1]{{\color{black} #1}}
\newcommand{\vthree}[1]{{\color{black} #1}}

\begin{document}

\preprint{APS/123-QED}

\title{Nonlocal growth of quantum conditional mutual information under decoherence}

%\title{Conditional Mutual Information under Decoherence and Connection to Many-body Teleportation}% Force line breaks with \\

\author{Yifan Zhang}
\email{yz4281@princeton.edu}
 
\author{Sarang Gopalakrishnan}%

\affiliation{Department of Electrical and Computer Engineering, Princeton University, Princeton, NJ 08544}
\date{\today}% It is always \today, today,

\begin{abstract}

Local measurements cannot create entanglement, but they can convert short-range entanglement to long-range entanglement, as in quantum teleportation. This phenomenon of measurement-induced entanglement (MIE) has been widely discussed in recent work on measurement-induced entanglement phase transitions and related phenomena. Here, we situate MIE in a broader context of the growth of long-range conditional mutual information (CMI) under decoherence. We upper-bound the rate at which decoherence can generate long-range CMI, and derive a characterization of states that saturate this bound. We point out that the structure of states saturating the CMI upper bound can be very different under different decoherent dynamics and provide explicit examples. We additionally explore the dynamics of CMI in random quantum circuits subject to random local decoherence, as a function of circuit depth. We argue that the universality class of the finite-depth teleportation transition, and its lower critical dimension, are different for erasures than for measurements.

%, interpolating between the finite-depth teleportation transition and the Page transition at infinite depth, which we explain using the statistical mechanics model.
%We argue that for erasure errors, the finite-depth teleportation transition only occurs in dimensions $d \geq 3$, by an Imry-Ma argument.

\end{abstract}

\maketitle

\section{Introduction}

Local measurements in a region cannot create entanglement between the region and its surroundings but can modify the structure of entanglement. For example, if Bob shares a Bell pair with Alice and another with Charlie, he can perform a Bell measurement on his qubits and create entanglement between Alice and Charlie, a protocol known as Bell teleportation \cite{nielsen_chuang}. Similar phenomena exist in large quantum systems, where a single layer of two-qubit measurements can generate long-range correlations \cite{PhysRevLett.81.5932,PhysRevLett.86.5188}. This nonlocal effect of measurements contrasts with the spreading of correlations in either local unitary circuits or circuits made up of local quantum channels---in the latter case, a Lieb-Robinson bound restricts correlations to grow inside a light cone \cite{cmp/1103858407,10.1103/physrevlett.97.050401}, while the corresponding bounds for circuits with measurements are much weaker \cite{10.48550/arxiv.2206.09929}. Because of the nonlocal effects of measurements, monitored systems can exhibit striking phenomena such as finite-time teleportation transitions~\cite{10.48550/arXiv.2110.06963}. 
This phenomenon of ``measurement-induced entanglement'' (MIE), by which measurements transmute short-range entanglement into long-range entanglement, has been explored in a number of contexts lately---for example, as a way of characterizing the structure of many-body quantum states \cite{10.22331/q-2023-02-02-910,10.48550/arXiv.2307.02292,10.1103/physreva.71.042306,10.1103/physrevlett.92.027901,10.1103/physrevlett.92.087201,10.48550/arXiv.2312.11615}. % and in the context of a finite-time phase transition in the structure of long-range entanglement in random quantum circuits \cite{10.48550/arXiv.2110.06963}. 

The nonlocal effects of measurements might seem paradoxical, given that the entire process of measuring a system and recording the measurement outcome can be regarded as a quantum channel. However, the sign of measurement-induced correlations is outcome-dependent, so averaging over outcomes washes out the correlations. To use the nonlocal correlations generated by measurement, one needs nonlocal classical communication and feedback (as in teleportation) or post-selection on particular measurement outcomes \cite{10.1103/physrevlett.127.220503,10.48550/arxiv.2206.13527,10.48550/arxiv.2205.01933,10.48550/arXiv.2208.11699,PRXQuantum.4.020315,10.48550/arxiv.2112.01519,2112.03061,10.48550/arxiv.2209.06202,zhu2023nishimori,chen2023realizing}. MIE is usually defined in a state that is conditioned on a set of measurement outcomes---and is thus defined as a property of individual \emph{trajectories} of the measurement channel, rather than an intrinsic property of the channel \vthree{(See Appendix \ref{teleport_metric} for the precise definition of MIE).} However, the fact that long-range correlations are generated regardless of the outcome suggests that there might be some intrinsic property of the channel itself that is responsible for creating them, regardless of how it is unraveled into trajectories. Characterizing MIE in terms of properties of a quantum channel would have the added benefit of making it possible to extend ideas about MIE to more general forms of interaction between a system and its environment.

The present work develops the quantum \emph{conditional mutual information} (CMI) as a diagnostic of such long-range correlations. We show that CMI (like MIE) can increase nonlocally as a result of decoherence, and relate it to MIE in the special case of the measurement channel. We discuss how the CMI evolves under general quantum channels, providing bounds both for general quantum channels and for the specific case of measurements. In addition, we describe the structure of quantum states for which the rate of CMI growth is maximal; we call these states ``efficient teleporters'', generalizing the concept of teleportation to generic quantum channels. For measurement channels, we identify a class of states---related to the Hayden-Preskill black hole paradox \cite{10.1088/1126-6708/2007/09/120}---for which measuring \emph{any} qubits gives rise to maximal CMI growth. On the other hand, we show that similar properties cannot be present in states subject to erasure channels, highlighting how the teleportation properties of different channels can be vastly different. Finally, we explore the rate at which quantum channels induce CMI growth in states generated by random quantum circuits as a function of circuit depth, interpolating between the teleportation transition and a ``Page'' transition in the limit of deep circuits. We numerically probe these transitions and the interpolation between them. We argue that the teleportation transitions for measurement and erasure channels lie in distinct universality classes.

This work is organized as follows. In Sec.~\ref{cmi} we introduce the basic properties of the classical and quantum CMI, including the relation between CMI and MIE. In Sec.~\ref{structure} we bound the growth of CMI for a single general decoherence event as well as for a single instance of a measurement channel, and characterize the structure of the states that saturate these bounds. In Sec.~\ref{hayden} we discuss examples of efficient teleporters and elaborate on the difference between erasure channels and measurement channels. Finally, in Sec.~\ref{ruc} we explore transitions in how many measurements one needs to create long-range correlations in a quantum circuit, as a function of circuit depth.

\vthree{
\section{Overview of Results}
We provide an overview of the results presented in this manuscript. We consider a tripartite system $ABC$, where $B$ is subject to a channel $\mathcal{N}[\cdot]$. After the channel, we observe that the difference in the CMI before and after the channel $\delta I = I(A:C|\mathcal{N}[B])-I(A:C|B)$ can be either positive or negative. The precise definition of CMI can be found in the next section. 

We focus on the situation where CMI grows under the channel, in other words $\delta I$ is positive. Our objective is threefold: first, to establish general bounds on how much $\delta I$ can increase; second, to understand the types of states that can saturate these bounds; and third, to understand how $\delta I$ grows under typical dynamics. We obtain the following results:

\begin{enumerate}
    \item We prove an upper bound of $\delta I$ that depends on the channel $\mathcal{N}[\cdot]$ (Theorem \ref{miebound0}). When the channel is a pure quantum instrument, we obtain a tighter bound (Theorem \ref{miebound1}).
    \item When the initial state is pure and $\delta I$ saturates the upper bound, we provide a structural theorem characterizing the state (Theorem \ref{erasure_structure} for generic channels and Theorem \ref{thm:structural_theorem} for pure quantum instruments). On the other hand, we show that mixed states can also saturate the upper bound but violate the structural theorems (Section \ref{mixed_state}).
    \item We discuss "efficient teleporters": they can saturate the CMI upper bound by applying a single-qubit channel to \emph{any} qubit in $B$ (section \ref{hayden}). We give an example of efficient teleporter under qubit measurements, inspired by the Hayden-Preskill protocol \cite{10.1088/1126-6708/2007/09/120}. We also show that efficient teleporters under erasure channels have a more restricted structure (Theorem \ref{efficient_teleporter_erasure}).
    \item We discuss the growth of CMI in states generated by random circuits, subject to qubit measurements and erasures (section \ref{ruc}). In both cases, we reveal a critical depth above which the growth of $\delta I$ almost saturates the upper bound. This can be considered as a "teleportation" phase transition.
\end{enumerate}

We will primarily focus on qubit systems, although the information-theoretic results are general.
}

\section{Conditional mutual information}\label{cmi}

\subsection{Decohered Quantum CMI}\label{quantum_cmi}
In classical information theory, Conditional Mutual Information (CMI) \(I(X:Z|Y)\) characterizes the correlation between two random variables, $X$ and $Z$, given the knowledge of a third variable, $Y$. Classically, the CMI is simply the mutual information between $X$ and $Z$ in the conditional distribution $P(X,Z|Y)$. In terms of Shannon entropies $H(\cdot)$, one can write the CMI as $I(X:Z|Y) = H(XY) + H(YZ) - H(Z) - H(XYZ)$. As a mutual information with respect to a well-defined probability distribution, CMI obeys the data processing inequality: $I(\mathcal{N}[X]:Z|Y) \leq I(X:Z|Y)$, where $\mathcal{N}[X]$ is a Markov chain acting on $X$. 

There is no quantum analog of a conditional distribution. However, one can define quantum analogs of quantities like the CMI by the standard strategy of writing the classical quantity out in terms of Shannon entropies and replacing the Shannon entropies by Von Neumann entropies, defined as follows:
\vthree{
\begin{equation}
    S(\rho) = -Tr[\rho \log(\rho)]
\end{equation}
Where $\rho$ denotes the density matrix. We use the base-two logarithm throughout this manuscript.
}
Thus, the quantum CMI is \emph{defined} as $I(A:C|B) = S(AB) + S(BC) - S(B) - S(ABC)$. It obeys the same data processing inequalities as the classical CMI as a consequence of the strong subadditivity of quantum entropy \cite{10.1063/1.1666274}: \(I(\mathcal{N}[A]:C|B) \le I(A:C|B)\), where \(\mathcal{N}\) is now a quantum channel. Nevertheless, when the \emph{conditioning system} \(B\) undergoes decoherence, the CMI can either increase or decrease. We now provide concrete examples of each behavior. Both examples involve density matrices that are diagonal in the computational basis. Thus the fact that $I(A:C|B)$ can be either larger or smaller than $I(A:C|\mathcal{N}[B])$ is also true for classical probability distributions.

\textbf{1. Decrease in CMI} Consider a tripartite system with the following density matrix:
\begin{equation}
    \begin{split}
        \rho_{ABC} &= \frac{1}{4}\ket{0}\bra{0}_A \otimes \ket{0}\bra{0}_B \otimes \ket{0}\bra{0}_C\\
                   &+ \frac{1}{4}\ket{0}\bra{0}_A \otimes \ket{1}\bra{1}_B \otimes \ket{1}\bra{1}_C \\
                   &+ \frac{1}{4}\ket{1}\bra{1}_A \otimes \ket{1}\bra{1}_B \otimes \ket{0}\bra{0}_C \\
                   &+ \frac{1}{4}\ket{1}\bra{1}_A \otimes \ket{0}\bra{0}_B \otimes \ket{1}\bra{1}_C
    \end{split}
\end{equation}
Here, systems A and C are independent classical bits with equal probabilities of being 1 or 0, while system B encodes their parity. It is straightforward to verify that \(I(A:C|B) = 1\), as knowing the parity in B correlates A and C according to the parity constraint. If the channel \(\mathcal{N}\) discards the bit in B, then \(I(A:C|\mathcal{N}[B]) = I(A:C) = 0\), since A and C revert to being uncorrelated without the parity information. This example demonstrates how decohering B can lead to a decrease in CMI.

\textbf{2. Increase in CMI} Consider another tripartite system with this density matrix:
\begin{equation}
\begin{split}
    \rho_{ABC} &= \frac{1}{2}\ket{0}\bra{0}_A \otimes \ket{0}\bra{0}_B \otimes \ket{0}\bra{0}_C\\
               &+ \frac{1}{2}\ket{1}\bra{1}_A \otimes \ket{1}\bra{1}_B \otimes \ket{1}\bra{1}_C
\end{split}
\end{equation}
In this scenario, systems A, B, and C are classical bits that are simultaneously either all 0 or all 1 with equal probability. It is straightforward to verify that \(I(A:C|B) = 0\), as the knowledge of B's value completely determines A and C, leaving no additional correlation between them. However, if the channel \(\mathcal{N}\) discards B's bit, \(I(A:C|\mathcal{N}[B]) = I(A:C) = 1\). Without the value of B, A and C become classically correlated bits. This example illustrates how decohering B can lead to a increase in CMI.

We comment that quantum CMI is already widely used in literature to probe nonlocal data of highly entangled quantum states. For instance, the topological entanglement entropy is defined as CMI among three adjacent regions \cite{PhysRevLett.96.110404,PhysRevLett.96.110405}. 
The connection between, CMI, quantum Markov chains, and recovery channels \cite{10.1093/qmath/39.1.97,10.1063/1.1459754,PhysRevA.77.034101} also sees application in understanding the entanglement structure of Gibbs states \cite{10.1007/s00220-019-03485-6,PhysRevLett.124.220601}, extracting universal data in fixed point ground states \cite{10.1016/j.aop.2020.168164,PhysRevB.103.115150,PhysRevLett.126.141602}, as well as solving quantum marginal problem in many-body systems \cite{PhysRevX.11.021039,10.48550/arXiv.2109.11688}. Here, we are interested in the dynamical aspects of CMI, namely how it increases upon decohering the conditioning system which is largely unexplored in many-body systems. To start with, we connect CMI to the concept of MIE by rewriting it as a CMI when the conditioning system undergoes a pure quantum instrument.

\subsection{CMI under measurements}

We now discuss the effects on $I(A:C|B)$ of generalized measurements on $B$. We implement measurements by means of a channel called a pure quantum instrument, which works as follows: we enlarge $B$ to include an ancillary register $X$, which is initially disentangled from $B$. A unitary is applied on $B \otimes X$, and then a dephasing channel is applied to $X$. Suppose $X$ has a basis $\{ \ket{i} \}, i = 1 \ldots k$. Then the pure quantum instrument $\mathcal{N}$ acts as follows:
\vtwo{
\begin{equation}
\mathcal{N}[\rho_B] = \sum_{i = 1}^k K_i \rho_B K_i^\dag \otimes \ket{i}\bra{i}_X,
\end{equation}
where \(K_i\) are Kraus operators subject to the condition \(\sum_i K_i^\dag K_i =I\). The probability of observing outcome i is given by \(p_i=\mathrm{Tr}[K_i \rho_B K_i^\dag]\). 
Note that when acting on a pure state, a pure quantum instrument results in pure states after post-selection. We will only consider pure quantum instruments in this work; more general classes of instrument can be defined~\cite{wilde2013quantum} but are outside the scope of this work.}

After a pure quantum instrument, conditioning on $B \otimes X$ is equivalent to \emph{classically} conditioning on the measurement outcome, as the following observation shows:

\begin{proposition}
\label{mie_cmi}
Consider an initial tripartite pure state $\ket{\psi} \in ABC$. Suppose a pure quantum instrument $\mathcal{N}$ acts on some subset of $B$, giving a classical measurement outcome $i \in 1\ldots k$ that is recorded in an ancillary register $X$. Then the CMI $I(A:C|\mathcal{N}[B])$ is precisely the average (weighted with Born probabilities) of the measurement-induced mutual information:

\begin{equation}
I(A:C|\mathcal{N}[B]) = \sum_{i = 1}^k p_i I(A:C|i).
\end{equation}
Here, $I(A:C|i)$ is the mutual information of $\ket{\psi_i}$, the pure state on $ABC$ corresponding to measurement outcome $i$.
\end{proposition}

\begin{proof}
The channel \(\mathcal{N}\) decoheres \(B\) and appends a classical register \(X\) to store the measurement outcome. After applying \(\mathcal{N}\), the density matrix on $A\mathcal{N}[B]C$ is
\begin{equation}
\label{eq:density_matrix}
    \rho_{A\mathcal{N}[B]C}=\sum_{i=1}^{k} p_i \ket{i}\bra{i}_{X} \otimes (\ket{\psi_i} \bra{\psi_i})_{ABC}.
\end{equation}
\vtwo{where \(\ket{\psi_i}_{ABC} \propto K_i \ket{\psi}_{ABC}\) are the post-selected pure states.} Because of the classical register, this density matrix is an ensemble of states that are orthogonal on $BX$ (and therefore on $ABX, CBX$). The classical register makes these states perfectly distinguishable on any subsystem containing $X$, so it is simple to see that (for example) $S(ABX) = \sum_i p_i S(\rho_{ABX,i}) + S(X)$, and likewise with the other terms. Plugging this into the definition of CMI, we find that \(I(A:C|\mathcal{N}[B])=\sum_{i=1}^{k} p_i I(A:C|i)\), where $i$ denotes classical conditioning, i.e., $I(A:C|i)$ means $I(A:C)$ evaluated in $\ket{\psi_i}$. 
\end{proof}

The generalization of this result to mixed states is:

\begin{observation}
\label{thm:MIE_bound_mixed}
For a mixed state $\rho \in ABC$, the CMI after the application of a pure quantum instrument is the weighted average of the CMI in each post-measurement state $\rho_i = M_i \rho M_i^\dagger$:

\begin{equation}
I(A:C|\mathcal{N}[B]) = \sum_{i = 1}^k p_i I(A:C|B, i).
\end{equation}
\end{observation}
Thus, when a mixed state is measured, the CMI after measurement is the quantum CMI of the state $\rho_i$---conditional on the classical measurement outcome $i$---averaged with Born weights over outcomes $i$.

\section{Evolution of CMI under decoherence}\label{structure}

%We now consider the emergence of non-local correlations under generic quantum channels,

We now discuss 
%which we probe by observing 
the increase in CMI \(I(A:C|B)\) under generic quantum channels when the conditioning system \(B\) undergoes decoherence. Specifically, we consider the following question: under a given channel, what is the maximal amount of CMI growth possible, and what classes of initial states allow for the maximal CMI growth? 
Our main results are: (i)~the CMI growth is upper bounded by the amount of entropy that is added to the system's density matrix by acting with the channel, and (ii)~states that saturate this bound are those in which distinct eigenvectors of the post-measurement state correspond to distinct orthogonal projectors on \( AC \).

%As we will see, there exists an upper bound to the CMI growth controlled by the amount of entropy it injects into the system. 
%
%The states that 

%In the case of pure states maximizing the CMI growth, we show that the eigenvectors of \(\rho_{A\mathcal{N}[B]C}\) are orthogonal on \(AC\). In the case of measurement channels, maximizing the CMI growth implies that different post-selected states are orthogonal on \(AC\) as well.

\subsection{General channels and pure states}\label{generic_channel_pure}

In this section we discuss how CMI evolves when a quantum channel acts on the conditioning system. Specifically we will consider an initially pure state on the tripartite system $ABC$. Before decoherence, CMI is simply the mutual information $I(A:C|B)=I(A:C) = S(A) + S(C) - S(AC)$, since $S(A) = S(BC)$, $S(C) = S(AB)$, $S(AC) = S(B)$, and $S(ABC) = 0$ for a globally pure state. We are interested in how $I(A:C|B)$ changes when a quantum channel acts on $B$. The quantum channel can be dilated into an isometry from $B$ to an enlarged space, $BE$, where $E$ is an environment that one traces over to define the quantum channel. It will be useful to work in the dilated picture, i.e., in terms of the properties of the pure state in the enlarged space $ABCE$. In this enlarged space, the entropy of $B$ prior to the channel is precisely the entropy $S(BE)$ in the final state. Thus, the change $\delta I$ in the CMI after the measurement can be expressed as 
\begin{eqnarray}
\delta I & = & S(AB) + S(BC) - S(B) - S(E) \nonumber \\ && \quad - (S(A) + S(C) - S(AC)).
\end{eqnarray}
After some straightforward manipulation we can write
\begin{equation}\label{cmi0}
\delta I = 2 S(E) - \{ I(A:E) + I(B:E) + I(C:E) \}.
\end{equation}
We note two immediate consequences of Eq.~\eqref{cmi0}. First, from the non-negativity of mutual information and the definition of $S(E)$, we have

\begin{theorem}
\label{miebound0}
Suppose $ABC$ are initially in a pure state $\ket{\psi}$. A channel $\mathcal{N}$ is applied to subsystem $B$. Then the CMI gain $\delta I \equiv I(A:C|\mathcal{N}[B]) - I(A:C|B)$ obeys the bound $\delta I \leq 2S(A\mathcal{N}[B]C)$.
\end{theorem}
This upper bound is tight as it is saturated by the following example, visualized in Fig. \ref{fig:cmi_structure}: 
\vtwo{Let \(\ket{\psi}_{ABC}\) be a six-qubit absolutely maximally entangled (AME) state \cite{goyeneche2015absolutely}. An AME state has the maximal entanglement possible across any bipartitions. Let A contain one qubit, B contain three qubits, and C contain two qubits. Per definition of the AME state, \(S(AB)=2\), \(S(BC)=1\), \(S(B)=3\), and \(S(ABC)=0\), and hence \(I(A:C|B)=0\) initially. Now suppose the channel erases one qubit in \(B\) and resets it to a pure state \(\ket{0}_E\). After applying the channel, \(S(A\mathcal{N}[B])=3\), \(S(\mathcal{N}[B]C)=2\), \(S(\mathcal{N}[B])=2\), and \(S(A\mathcal{N}[B]C)=1\), and thus \(\delta I = 2\), saturating the upper bound. 

Note that while the AME property is fine-tuned and only exists in few-qubit systems \footnote{For example, in qubit systems AME states only exist for two, three, five, and six qubits \cite{huber2017absolutely}}, the above upper bound can also be almost saturated, up to an \(O(1)\) discrepancy and with high probability, by typical Haar-random states which are more generic and robust. To see that, take an \(N\)-qubit typical Haar-random state and take a tripartition with \(N_A\), \(N_B\), and \(N_C\) qubits in each partition such that \(N_B=N_A+N_C=N/2\). By ignoring the \(O(1)\) Page correction, \(S(AB)=N_C\), \(S(BC)=N_A\), \(S(B)=N_A+N_C\), and \(S(ABC)=0\), and hence \(I(A:C|B)=0\). By erasing \(M\) qubits in \(B\) satisfying \(M+\text{min}(N_A, N_C) < N/2\), \(S(A\mathcal{N}[B])=N_C+M\), \(S(\mathcal{N}[B]C)=N_A+M\), \(S(\mathcal{N}[B])=N_A+N_C-M\), and \(S(A\mathcal{B}C)=M\), and thus \(\delta I = 2 M\), up to the \(O(1)\) Page correction.}

The class of states for which $\delta I$ is maximal can be characterized as follows:

\begin{theorem}\label{erasure_structure}
  Given a pure state $\ket{\psi}_{ABC}$ and a pure quantum instrument $\mathcal{N}$ acting on $B$ that maximizes the CMI gain $\delta I = 2S(A\mathcal{N}[B]C)$. Let \(\sigma = \mathcal{N}[\ket{\psi}_{ABC}\bra{\psi}]\) be the state after applying the channel.
Then the eigenvectors $\ket{\phi_i}_{ABC}$ of the density matrix $\sigma$ of the post-decoherence system have the following property: 
\begin{equation}
    \mathrm{Tr}_B [\ket{\phi_i} _{ABC}\bra{\phi_i}]\mathrm{Tr}_B [\ket{\phi_j} _{ABC}\bra{\phi_j}]=0, \; \forall i \neq j
\end{equation}
In other words, different $\mathrm{Tr}_B [\ket{\phi_i} _{ABC}\bra{\phi_i}]$ are supported on orthogonal subspaces of $AC$.
% \begin{quote}
% $\forall i \neq j$, $\mathrm{Tr}_B [\ket{\phi_i} _{ABC}\bra{\phi_i}]\mathrm{Tr}_B [\ket{\phi_j} _{ABC}\bra{\phi_j}]=0$. In other words, $\mathrm{Tr}_B [\ket{\phi_i} _{ABC}\bra{\phi_i}]$ and $\mathrm{Tr}_B [\ket{\phi_j}_{ABC} \bra{\phi_j}]$ are supported on orthogonal subspaces of $AC$.
% \end{quote}
\end{theorem}
\vthree{
\begin{proof}
By Eq.~\eqref{cmi0}, \({I(A:E)}={I(B:E)}={I(C:E)}=0\). Therefore, the environment is entangled with the composite system but decoupled from each individual subsystem $A, B, C$. \({I(B:E)}=0\) is equivalent to $\rho_{BE}=\rho_{B} \otimes \rho_{E}$ which allows us to invoke the decoupling theorem.
    \begin{lemma}
 \label{decoupling}   \emph{Decoupling theorem} \cite{preskill1998lecture}.--- Given a pure state $\ket{\psi}_{XYZ}$. Suppose the reduced density matrix on $XZ$ is separable: $\rho_{XZ} = \rho_{X} \otimes \rho_{Z}$, then there exist a local unitary $W_Y$ acting on $Y$ such that
    \begin{equation}\label{decoupling_decomposition}
        W_Y \ket{\psi}_{XYZ}\bra{\psi} W_Y^\dag = \ket{\psi_1}_{XY_1}\bra{\psi_1} \otimes \ket{\psi_2}_{Y_2Z}\bra{\psi_2}
    \end{equation}
    where $Y_1Y_2$ is a bipartition of $Y$ and $\ket{\psi_1}$, $\ket{\psi_2}$ are some pure states.
\end{lemma}
We now set $X = E$, $Y=AC$, and $Z=B$. Since $E$ and $B$ are decoupled, one can rotate $AC$ such that 
    \begin{equation}
        W_{AC} \ket{\psi}_{ABEC}\bra{\psi} W_{AC}^\dag = \ket{\psi_1}_{E\{AC\}_1}\bra{\psi_1} \otimes \ket{\psi_2}_{\{AC\}_2 B}\bra{\psi_2}
    \end{equation}
where $\{AC\}_1$, $\{AC\}_2$ denotes a bipartition of $AC$. To find the eigenvalues $\ket{\phi_i} _{ABC}$, we use the Schmidt decomposition of $\ket{\psi_1}$.
\begin{align}
    \ket{\psi_1}_{E\{AC\}_1} = \sum_i \lambda_{i} \ket{e_{i}}_E \otimes \ket{a_{i}}_{\{AC\}_1}
\end{align}
where $\lambda_{i}$ are the Schmidt coefficients and $\ket{e_{i}}_E $ and $\ket{a_{i}}_{\{AC\}_1}$ are some orthonormal basis on $E$ and $\{AC\}_1$, respectively. Under these basis, the post-decoherence state $\sigma=Tr_E[\ket{\psi}_{ABCE}\bra{\psi}]$ can be written in the following diagonal form:
\begin{equation}
    \sigma = \sum_i |\lambda_{i}|^2 \ket{a_{i}}_{\{AC\}_1}\bra{a_{i}} \otimes \ket{\psi_2}_{\{AC\}_2 B}\bra{\psi_2}
\end{equation}
In this form, it becomes apparent that the eigenvectors $\ket{\phi_i}_{ABC}$ are exactly $\ket{a_{i}}_{\{AC\}_1} \otimes \ket{\psi_2}_{\{AC\}_2 B}$. Crucially, after tracing out $B$, the $\ket{a_{i}}_{\{AC\}_1}$ part remains intact.
\begin{equation}
    \mathrm{Tr}_B [\ket{\phi_i} _{ABC}\bra{\phi_i}] = \ket{a_{i}}_{\{AC\}_1}\bra{a_{i}} \otimes Tr_B[\ket{\psi_2}_{\{AC\}_2 B}\bra{\psi_2}]
\end{equation}
Thus, different $\ket{\phi_i} _{ABC}$ are orthogonal even after tracing out $B$ because of the orthogonality of $\ket{a_{i}}_{\{AC\}_1}$.
\end{proof}
}
%
% \begin{proof}
% By Eq.~\eqref{cmi0}, \({I(A:E)}={I(B:E)}={I(C:E)}=0\). Therefore, the environment is entangled with the composite system but decoupled from each individual subsystem $A, B, C$. Since \({I(B:E)}=0\), the decoupling theorem implies that $AC$ purifies $E$: i.e., there exists a unitary transformation on $AC$ that separates out some degrees of freedom that are entangled with $B$ and others that are entangled with $E$ (see e.g. \cite{nielsen_chuang}). The observation then follows from basic properties of the Schmidt decomposition of $\ket{\psi}$ between $AC$ and $E$. 
% \end{proof}
%
When Eq.~\eqref{cmi0} is maximized, the state on $ABCE$ can be interpreted as a quantum error correcting code: $E$ acts as a reference qubit for a code-space that is isometrically embedded in $ABC$, such that the information is protected against the erasure of any one of these subsystems.

\subsection{Measurement channels and pure states}\label{measurement_channel_pure}

We now consider the dynamics of CMI under the application of a pure quantum instrument. Again we restrict to initial states being pure and will lift this restriction in the next section. One can immediately see from~\eqref{cmi0} that the largest possible CMI gain is now $S(E)$ rather than $2S(E)$, since $I(B:E) \geq S(E)$ for any pure quantum instrument.

\begin{theorem}
   \label{miebound1} Suppose $ABC$ are initially in a pure state $\ket{\psi}$. A pure quantum instrument $\mathcal{N}$ is applied to subsystem $B$. Then the CMI gain $\delta I \equiv I(A:C|\mathcal{N}[B]) - I(A:C|B)$ obeys the bound $\delta I \leq S(A\mathcal{N}[B]C)$. Moreover, \(S(A\mathcal{N}[B]C)\) is the same as the Shannon entropy of the measurement outcome distribution.
\end{theorem}
\begin{proof}
 $\delta I \leq S(A\mathcal{N}[B]C)$ follows trivially from \(I(B:E) \ge S(E)\) which is true for pure quantum instruments because one can always copy the information in the classical register and send it to the environment. (\ref{eq:density_matrix}) gives the structure of the post-measurement state as an mixture of orthogonal post-selected pure states at the probability of the measurement outcome \(p_i\). 
 \vtwo{Specifically, the classical registers guarantee the orthogonality: \(\braket{i|j}_X = \delta_{ij}\). \(S(A\mathcal{N}[B]C)\) is then equal to:
 \begin{equation}
     S(A\mathcal{N}[B]C) = \sum_i p_i S(A\mathcal{N}[B]C|i) -\sum_i  p_i \log p_i
 \end{equation}
 Since the post-selected states are pure, \(S(A\mathcal{N}[B]C|i)=0, \forall i\), so}
 \(S(A\mathcal{N}[B]C)=-\sum_i  p_i \log p_i\) which coincides with the Shannon entropy of the measurement outcome distribution.
\end{proof}
Again we know the upper bound is tight because Bell teleportation saturates this bound. Visualized in Fig. \ref{fig:cmi_structure}(b), Bell teleportation consists of a Bell state \(1/\sqrt{2}(\ket{00}+\ket{11})\) shared between \(A\) and \(B\), and another Bell state shared between \(B\) and \(C\). \(B\) is then measured in the Bell basis. The four measurement outcomes are equally probable, so \(S(A\mathcal{N}[B]C)=2\). On the other hand, the resulting state is always maximally entangled between \(A\) and \(C\), so \(I(A:C|i)=2, \forall i\), thus the upper bound is saturated. As an immediate corollary, if the pure quantum instrument measures \(N\) qubits in \(B\), then \(\delta I \le S(A\mathcal{N}[B]C) \le N\). In other words, measuring \(N\) qubits can give rise to at most \(N\) bits of CMI growth.

%\vtwo{Finally, notice that the above result is valid only when the quantum instrument is pure. To give a counterexample, consider the following instrument that returns a classical register \(\ket{0}_X\) deterministically and apply an erasure channel on \(B\). This instrument has a zero classical Shannon entropy, yet it could increase the CMI as we have seen in the previous section because of the erasure. The above theorem fails because \(S(A\mathcal{N}[B]C|i)\) are no longer zero when the instrument is mixed.}
%

Unlike generic channels, the environment of the measurement channel is now classically correlated with the system, so the decoupling theorem is no longer available for characterizing states with maximal CMI gain. Instead, we will have to follow a more involved proof strategy based on accessible information. The result, however, closely parallels  We state the result here, and defer the proof to Appendix \ref{proof}.

\begin{figure}
\scalebox{0.4}{\includegraphics{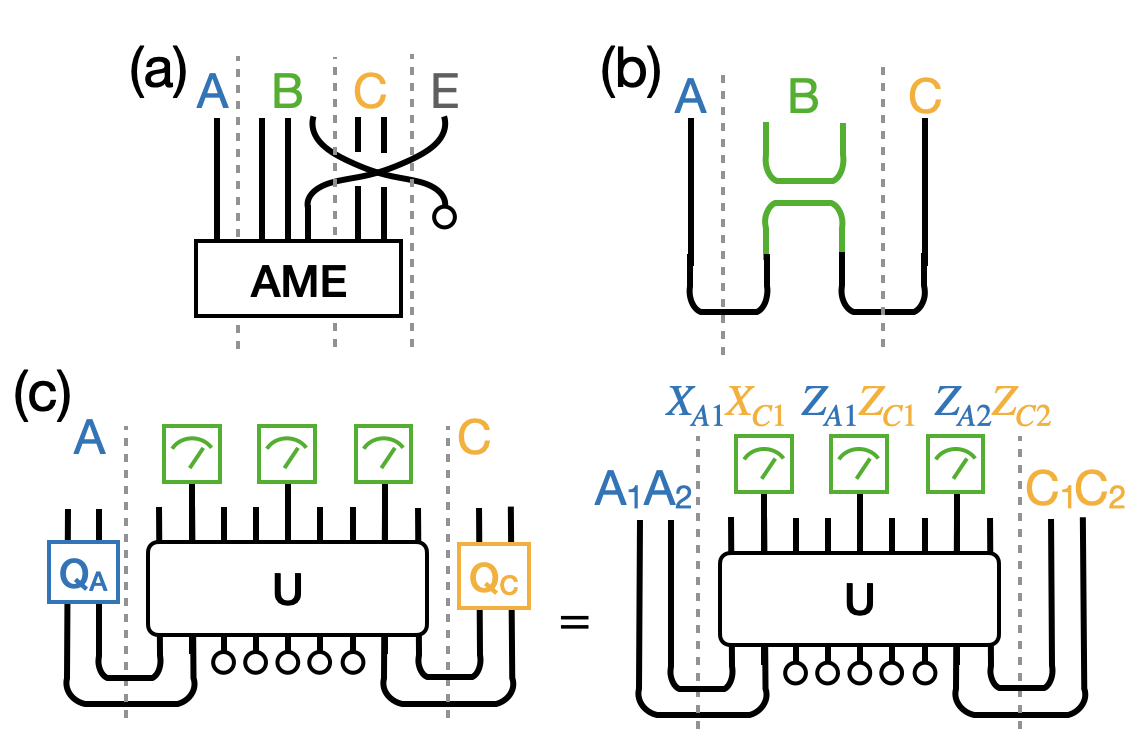}}
\caption{\label{fig:cmi_structure}(a) Example that saturates the CMI upper bound for generic channel. (b) Example that saturates the CMI upper bound for measurement channel . (c) Given a state saturating the CMI upper bound for measurement channel, if \(\rho_{AC}\) is maximally mixed, then the measurement protocol is local-unitary equivalent to measuring Bell stabilizers.}
\end{figure}

%Lastly, we provide a structural theorem for efficient teleporters. It states that different post-selected density matrices have to stay in orthogonal subspaces. Further, under some additional conditions all efficeint teleporters are local unitary-equivalent to Bell teleportation.
\begin{theorem}
 \label{thm:structural_theorem}
    Given a pure state $\ket{\psi}_{ABC}$ and a pure quantum instrument $\mathcal{N}$ acting on $B$ that maximizes the CMI gain $\delta I = S(A\mathcal{N}[B]C) = -\sum_i p_i \log p_i$. Let $\rho_{AC|i}$ be the reduced density matrix on $AC$ corresponding to measurement outcome $i$. Then, 
    \begin{equation}
        \rho_{AC|i}\rho_{AC|j} = 0, \, \forall i \neq j
    \end{equation}
In other words, different \(\rho_{AC|i}\) are supported on orthogonal subspaces of \(AC\).
    
  %  efficient teleporter with \(I_T=N\), let \(\rho_{AC|i}=Tr_{\bar{B}}[\ket{\psi_i}\bra{\psi_i}_{A\bar{B}C}]\), where \(\ket{\psi_i}_{A\bar{B}C}\)  is the post-selected state defined previously. Then, \(\forall i \neq j \), \(\rho_{AC|i}\rho_{AC|j}=0\).

    In addition, if  \(\rho_{AC} \) is maximally mixed and suppose that the pure quantum instrument measures \(N\) qubits in \(B\) giving equally probable outcomes, then there exist local unitaries \(Q_A, Q_C\) acting on A and C such that after the application the each qubit measures some Bell stabilizer in \(AC\). 
\end{theorem}
%The proof is technical so we defer it to []. 

Fig. \ref{fig:cmi_structure}(c) visualizes the second part of the Observation \ref{thm:structural_theorem}. Under the assumption that \(\rho_{AC}\propto\) identity matrix, one can perform local rotations such that \(B_1,B_2,...B_N\) are measuring some Bell stabilizers such as \(X_{A1}X_{C1}\), \(Z_{A2}Z_{C2}\), as shown in the figure. This observation complete characterizes system where MIE is maximal. In general, the post-measurement density matrices on \(AC\) are orthogonal and thus maximally distinguishable. Under some additional constrains, all systems achieving maximal MIE growth are equivalent to Bell teleportation.

\subsection{Case of mixed states}\label{mixed_state}

We now discuss how the results of the previous sections extend to mixed states. In this case and for both generic channels and the measurement channels, the CMI upper-bound still holds true, but the structural theorem no longer holds.

\vtwo{
\begin{theorem}
The upper bound \(\delta I \le 2S(E)\) for generic channels and \(\delta I \le S(E)\) for pure quantum instruments still holds true even if the initial state \(\rho_{ABC}\) is mixed.
\end{theorem}
\begin{proof}
    We purify the state $\rho_{ABC}$ by adding a reference system $R$ such that $\ket{\psi}_{ABCR}$ is a purification of $\rho_{ABC}$. Dilating the channel to explicitly include an environment $E$, as before, we can write the state following the application of the channel $\mathcal{N}$ as $\ket{\Psi}_{ABCRE}$. With reference to $\ket{\Psi}_{ABCRE}$, the CMI change is $\delta I = I(A:C|B) - I(A:C|BE)$. We now add and subtract $I(A:C|BER)$ to rearrange this expression as:
    \begin{eqnarray}\label{cmimix}
    \delta I & = & [I(A:C|B) - I(A:C|BER)] \nonumber \\ & & \quad - [ I(A:C|BE) - I(A:C|BER) ].
    \end{eqnarray}
Each of the expressions in square brackets refers to the change in CMI when tracing out part of the conditioned subsystem in a globally pure state. This observation allows us to use Eq.~\eqref{cmi0} to express~\eqref{cmimix} as follows:
\begin{eqnarray}
\delta I & = & 2 (S(ER) - S(R)) + (I(BE:R) - I(B:ER)) \\ & & + I(A:R) - I(A:ER) + I(C:R) - I(C:ER).\nonumber 
\end{eqnarray}
The second line is upper-bounded by zero by the quantum data processing inequality. Plugging in the definition of mutual information, one finds that
\begin{equation}
\delta I \leq S(ER) - S(R) + S(EB) - S(B) \leq 2 S(E),
\end{equation}
by subadditivity of entanglement entropy. In the case of measurement channels, \(I(E:B) \ge S_E\) is equivalent to \(S_{BE}-S_{B} \le 0\), so \(\delta I \le S_E\).
    
    % We introduce the purification \(R\) of \(ABC\) before the channel is applied such that the global state is \(\ket{\psi}_{ABCR}\). In general, \(I(A:C|RB)\) can be bigger or small than \(I(A:C|B)\), but we can use Eqs. (\ref{cmi0}) to calculate \(I(A:C|B)-I(A:C|RB)\) given the global purity.
    % \begin{equation}\label{eq:mix_diff_1}
    % \begin{split}
    %     I(A:C|B)-I(A:C|RB) = \\
    %     2 S(R) - I(A:R) - I(B:R) - I(C:R)
    % \end{split}
    % \end{equation}
    % But we are interested in \(I(A:C|\mathcal{N}[B])\), so we evaluate \(I(A:C|\mathcal{N}[B])-I(A:C|RB)\) and exploit the global purity again.
    % \begin{equation}\label{eq:mix_diff_2}
    % \begin{split}
    %     I(A:C|\mathcal{N}[B])-I(A:C|RB) =\\
    %     2 S(RE) - I(A:RE) - I(\mathcal{N}[B]:RE) - I(C:RE)
    % \end{split}
    % \end{equation}
    % Notice that \(B\) in Eqs. (\ref{eq:mix_diff_1}) is equivalent to \(\mathcal{N}[B]E\) in Eqs. (\ref{eq:mix_diff_2}). Taking (\ref{eq:mix_diff_2}) - (\ref{eq:mix_diff_1}) gives \(\delta I\) on the left hand side. On the right hand side, \(I(A:RE)-I(A:R) \ge 0\) and \(I(C:RE)-I(C:R) \ge 0\) by the data processing inequality. Thus,
    % \begin{align}
    % \begin{split}
    %     \delta I \le& 2(S(RE)-S(R)) \\
    %     &- (I(\mathcal{N}[B]:RE) - I(\mathcal{N}[B]R:E))
    % \end{split}\\
    %     =&(S(RE)-S(R))+(S(\mathcal{N}[B]E)-S(\mathcal{N}[B]))
    % \end{align}
    % \(S_{ER}-S{R} \le S_E\) and  \(S_{BE}-S_{B} \le S_E\) for generic channels, resulting in \(\delta I \le 2 S_E\). In the case of measurement channels, \(I(E:B) \ge S_E\) is equivalent to \(S_{BE}-S_{B} \le 0\), so \(\delta I \le S_E\).
\end{proof}
}
Thus, the upper bounds on $\delta I$ remain valid even if the initial state is mixed. However, maximal CMI gain need not imply anything about the structure of the initial state.

\begin{observation}
For mixed initial states satisfying the upper bound on $\delta I$, it need not be the case that the eigenvectors of the post-channel state are distinguishable on $AC$. In other words, Theorem \ref{erasure_structure} and \ref{thm:structural_theorem} does not hold for mixed initial states.
\end{observation}
To substantiate this observation we provide two counterexamples:

%However, the structural theorem is no longer true as we will give two counter-example, one for erasure channel and one for measurement channel.

\vtwo{
\textbf{Erasure Channel} This example is visualized in Fig. \ref{fig:mixed_counterexample}. Consider a six-qubit AME state partitioned into \(ABCR\) such that $A$ contains one qubit, $B$ contains three qubits, $C$ contains one qubit, and $R$ contains one qubit. We trace out \(R\) so that the resulting state on $ABC$ is mixed. One can easily verify that \(S(AB)=2\), \(S(BC)=2\), \(S(B)=3\), \(S(ABC)=1\), so \(I(A:C|B)=0\) initially. We now erase one qubit in $B$. After that, \(S(A\mathcal{N}[B])=3\), \(S(\mathcal{N}[B]C)=3\), \(S(\mathcal{N}[B])=2\), \(S(A\mathcal{N}[B]C)=2\), so \(I(A:C|\mathcal{N}[B])=2\), saturating the CMI upper bound. 

Nevertheless, we now show 
that 
%structural theorem regarding 
the orthogonality relation no longer holds. Let $B_1$ be the qubit in $B$ to be erased and let $B_2B_3$ be the other two qubits. Following \cite{helwig2012absolute}, the six-qubit AME state can be written as
\begin{equation}
    \ket{\psi}_{ABCR} = \frac{1}{\sqrt{2}}(\ket{0}_{B_1}\ket{0_L}_{AB_2B_3CR}+\ket{1}_{B_1}\ket{1_L}_{AB_2B_3CR})
\end{equation}
Where $\ket{0_L}_{AB_2B_3CR}$ and $\ket{1_L}_{AB_2B_3CR}$ are the logical 0 and 1 states of the five-qubit error correcting code. Crucially, $\ket{0_L}_{AB_2B_3CR}$ and $\ket{1_L}_{AB_2B_3CR}$ are also five-qubit AME states. One can see, either from this AME property or because the code distance of the 5-qubit code is 3, that the two states $\ket{0_L}_{AB_2B_3CR}$ and $\ket{1_L}_{AB_2B_3CR}$ give \emph{identical} reduced density matrices on $AC$---namely, the maximally mixed state. Thus the orthogonality relation clearly fails.

%The structural theorem is about distinguishing $\rho_{AC}$ when post-selecting $\rho_{B_1}$, so we evaluate $\rho_{B_1 AC}$:
%\begin{equation}
%\begin{split}
 %   \rho_{B_1 AC} =& \frac{1}{2}\ket{0}\bra{0}_{B_1}Tr_{R B_2B_3}[\ket{0_L}\bra{0_L}_{AB_2B_3CR}]\\
  %  &+\frac{1}{2}\ket{1}\bra{1}_{B_1}Tr_{R B_2B_3}[\ket{1_L}\bra{1_L}_{AB_2B_3CR}]
%\end{split}
%\end{equation}
%Notice that there are no off-diagonal elements because $\ket{0_L}_{AB_2B_3CR}$ and $\ket{1_L}_{AB_2B_3CR}$ are orthogonal code words. If the structural theorem holds, one would expect that $Tr_{R B_2B_3}[\ket{0_L}\bra{0_L}_{AB_2B_3CR}]$ and $Tr_{R B_2B_3}[\ket{0_L}\bra{0_L}_{AB_2B_3CR}]$ are supported on orthogonal subspace, but in fact they are both maximally mixed states because $\ket{0_L}_{AB_2B_3CR}$ and $\ket{1_L}_{AB_2B_3CR}$ are AME states as well. Therefore, the structural theorem no longer holds.
}

% \textbf{Erasure Channel} This example is visualized in Fig. \ref{fig:mixed_counterexample}. Suppose A has one qubit, B has three qubits \(B_1,B_2,B_3\), and C has one qubit. In addition, we introduce a reference state R with one qubit. We prepare a Haar-random state \(\ket{\psi}_{AB_1B_2B_3CR}\) where again we fine tune such that page correction vanishes. Now the initial state \(\rho_{ABC}\) is mixed, but \(S(AB)=2\), \(S(BC)=2\), \(S(B)=3\), \(S(ABC)=1\), so \(I(A:C|B)=0\). Again, we erase \(B_1\) and replace with a clean qubit \(\ket{0}_E\) so that the resulting state is \(\ket{\psi}_{AEB_2B_3CR} \otimes \ket{0}_{B_1}\). Now, \(S(A\mathcal{N}[B])=3\), \(S(\mathcal{N}[B]C)=3\), \(S(\mathcal{N}[B])=2\), \(S(A\mathcal{N}[B]C)=2\), so \(I(A:C|\mathcal{N}[B])=2\), saturating the CMI upper bound. However, suppose we measure E and obtain two possible outcomes \(\braket{0_E|\psi}_{AEB_2B_3CR} \otimes \ket{0}_{B_1}\) and \(\braket{1_E|\psi}_{AEB_2B_3CR} \otimes \ket{0}_{B_1}\). Note that since E is maximally entangled with the rest of the system, it doesn't matter which basis is measured--any basis form a Schmidt decomposition. It is known that measuring and post-selecting a Haar-random gives a smaller Haar-random state. In either post-selection, \(S(AC)=2\) so \(\rho_{AC|0}\) and \(\rho_{AC|1}\) are both maximally mixed, and thus are not orthogonal to each other.

\begin{figure}
\scalebox{0.4}{\includegraphics{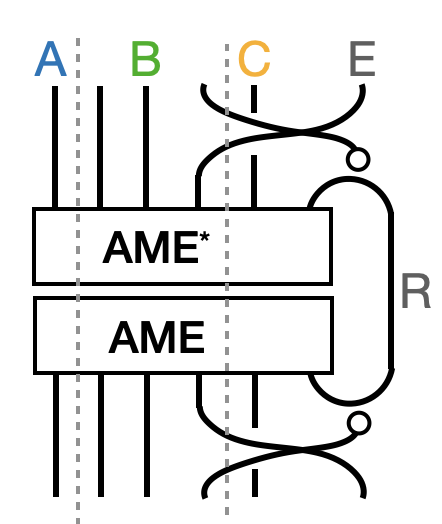}}
\caption{\label{fig:mixed_counterexample}The counterexample where a mixed state saturates the CMI upper bound under the erasure channel but does not have the orthogonality property}
\end{figure}

\textbf{Measurement Channel} consider the following mixed state with \(I(A:C|B)=0\).
\begin{equation}
\begin{split}
    \rho_{ABC} &= \frac{1}{2}\ket{0}\bra{0}_A \otimes \ket{0}\bra{0}_B \otimes \ket{0}\bra{0}_C\\
               &+ \frac{1}{2}\ket{1}\bra{1}_A \otimes \ket{1}\bra{1}_B \otimes \ket{1}\bra{1}_C
\end{split}
\end{equation}
This is the same mixed state as discussed in section \ref{quantum_cmi}. Now the measurement channel measures B in the X basis, so that after the measurement channel the state becomes
\begin{equation}
\begin{split}
    \rho_{A\mathcal{N}[B]C} &= \frac{1}{2}(\ket{+x}\bra{+x}_{B}+\ket{-x}\bra{-x}_{B})\\
    &\otimes(\ket{0}\bra{0}_A \otimes \ket{0}\bra{0}_C+ \frac{1}{2}\ket{1}\bra{1}_A \otimes \ket{1}\bra{1}_C)
\end{split}
\end{equation}
Where we ignore the classical registers since the post-measurement states in \(B\) is already orthogonal. It is easy to see that \(I(A:C|\mathcal{N}[B])=1\), saturating the upper bound. Nevertheless, no matter what is the measurement outcome, \(\rho_{AC|i}\) is always \(\ket{0}\bra{0}_A \otimes \ket{0}\bra{0}_C+ \frac{1}{2}\ket{1}\bra{1}_A \otimes \ket{1}\bra{1}_C\), so the structural theorem no longer holds as well.

\section{Efficient teleporters}\label{hayden}

In this section we discuss efficient teleporters: many-body examples of states with maximal CMI growth under decoherence. Specifically, we consider the case when \(\delta I=2S(A\mathcal{N}[B]C))\) for generic channels and \(\delta I=S(A\mathcal{N}[B]C)\) for measurement channels, and moreover, demand that \(S(A\mathcal{N}[B]C)=S(E)=N_E\). In other words, not only do we want the CMI to saturate the upper bound as controlled by the entropy injected, we require the entropy to be maximal as allowed by the channel. 

We will consider the following setup: an initial pure state is prepared on $ABC$, which respectively contain $N_A, N_B, N_C$ qubits. One of the qubits in $B$ is randomly chosen and subjected to a channel. Under what condition can the CMI growth be maximal? As we will see, states that saturate the general bound (Theorem \ref{miebound0}) are much more restricted than those that saturate the bound for the measurement channel (Theorem \ref{miebound1}). 
% Specifically, we argue that efficient teleportation is impossible under erasure channel if \(N_B \geq min(N_A,N_C)\). In the case of measurement channel we give a class of states, a Clifford version of the Hayden-Preskill protocol, where the CMI upper bound can be saturated even if \(N_B \gg N_A,N_C\).

\subsection{Efficient teleporters under erasure}
We can construct a class of efficient teleporters under erasure (up to an \(O(1)\) discrepancy) using the following procedure: pair each qubit in $AC$ with a qubit in $B$ and prepare each pair in a maximally entangled state, then apply a random unitary rotation on $AC$. The resulting state is shown in Fig. \ref{fig:efficient_eraser}(a). The six-qubit AME state example in \ref{generic_channel_pure} can also be written in this form. In this case, any one of the qubits in $B$ can be erased giving the maximum $\delta I$, in the limit where all subsystems have $N \to \infty$. Obviously, \(N_B=N_A+N_C\) in order to pair each qubit in $AC$ with a qubit in $B$. In other words, the relative size between \(AC\) and \(B\) has to be fixed to produce efficient teleporters with this strategy. 

One might ask if it is possible to do better: e.g., if any one of the qubits in $N_B$ can be measured giving maximum growth in $\delta I$, even if the size of \(B\) is bigger than \(AC\), for example in the setup shown in Fig. \ref{fig:efficient_eraser}(b). It turns out this is impossible.

\begin{theorem}\label{efficient_teleporter_erasure}
Given a state \(\ket{\psi}_{ABC}\), at most \(N_A+N_C\) qubits in \(B\) can have the property that erasing any one of them gives rise to maximal CMI growth of two.
\end{theorem}
\vtwo{
\begin{proof}
\vthree{For concreteness let us first consider the case $k = 2$. Suppose $B = B_0 Q_1 Q_2$, where either $Q_1$ or $Q_2$ can be erased (i.e., transferred from $B$ to $E$) giving maximal $\delta I$. Recall that to saturate the CMI upper bound, we need \(I(E:B)=0\). Therefore, we have $I(B_0 Q_1: Q_2) = 0$ and $I(B_0 Q_2:Q_1) = 0$. This implies the following factorizations of the density matrix:
\begin{equation}
    \rho_{B} = \rho_{B_0 Q_1} \otimes \rho_{Q_2} = \rho_{B_0 Q_2} \otimes \rho_{Q_1}
\end{equation}
by data-processing inequality, $I(B_0 Q_1: Q_2) = 0$ implies that $I(B_0: Q_2) = 0$, so that $\rho_{B_0 Q_2}=\rho_{B_0} \otimes \rho_{Q_2}$. Thus,
\begin{equation}
    \rho_{B} = \rho_{B_0} \otimes \rho_{Q_1} \otimes \rho_{Q_2}
\end{equation}
The above argument can be generalized to to erasing $k$ qubits $Q_1 Q_2 \ldots Q_k$: if each $k$ can be erased with maximal $\delta I$, then $\rho_{B} = \rho_{B_0} \otimes \rho_{Q_1} \otimes \ldots \otimes \rho_{Q_k}$. Because of the factorization, we invoke the decoupling theorem (Lemma \ref{decoupling}). It is obvious that Eq. (\ref{decoupling_decomposition}) implies the following relation about the coherent information
    \begin{equation}
        I(X \rangle Y) = I(X \rangle YZ) = S(X)
    \end{equation}
    Where $I(X \rangle Y) \equiv S(Y) - S(XY)$ denotes the coherent information. The above equation is essentially the decoupling theorem stated from a communication perspective: as long as $X$ and $Z$ are decoupled, quantum information purified by $X$ can be perfectly recovered in $Y$. Applying the above relation and setting $X=Q_1Q_2 \ldots Q_k$, $Y=B_0$, and $Z=AC$, we have 
\begin{equation}
    I(Q_1 \ldots Q_k \rangle AC) = S(Q_1 \ldots Q_k),
\end{equation}
Since the coherent information is upper bounded by $N_A + N_C$, it follows that \(S(Q_1 \ldots Q_k) \le N_A+N_C\). Given that $\rho_{B} = \rho_{B_0} \rho_{Q_1} \ldots \rho_{Q_k}$, and each $\rho_{Q_k}$ has the maximal possible entropy of one bit, $S(Q_1 \ldots Q_k) = \sum_{i = 1}^k S(Q_i) = k$ (in bits). It follows that \(k \le N_A+N_C\).
}

\end{proof}
}

We note what while at most \(N_A+N_C\) qubits can give rise to maximal CMI growth upon erasure, it does not imply that erasing all of them together increases CMI maximally as well. This can already seen in the state depicted in Fig. \ref{fig:efficient_eraser}(a). Erasing any one qubit in \(B\) gives maximal CMI growth, but if all qubits in \(B\) are erased, then CMI decreases to 0. In fact, it is easy to see that at most \(\min(N_A,N_C)\) qubits can be erased together to give maximal CMI growth since CMI is upper bounded by \(2\min(N_A,N_C)\).
\begin{figure}
\scalebox{0.35}{\includegraphics{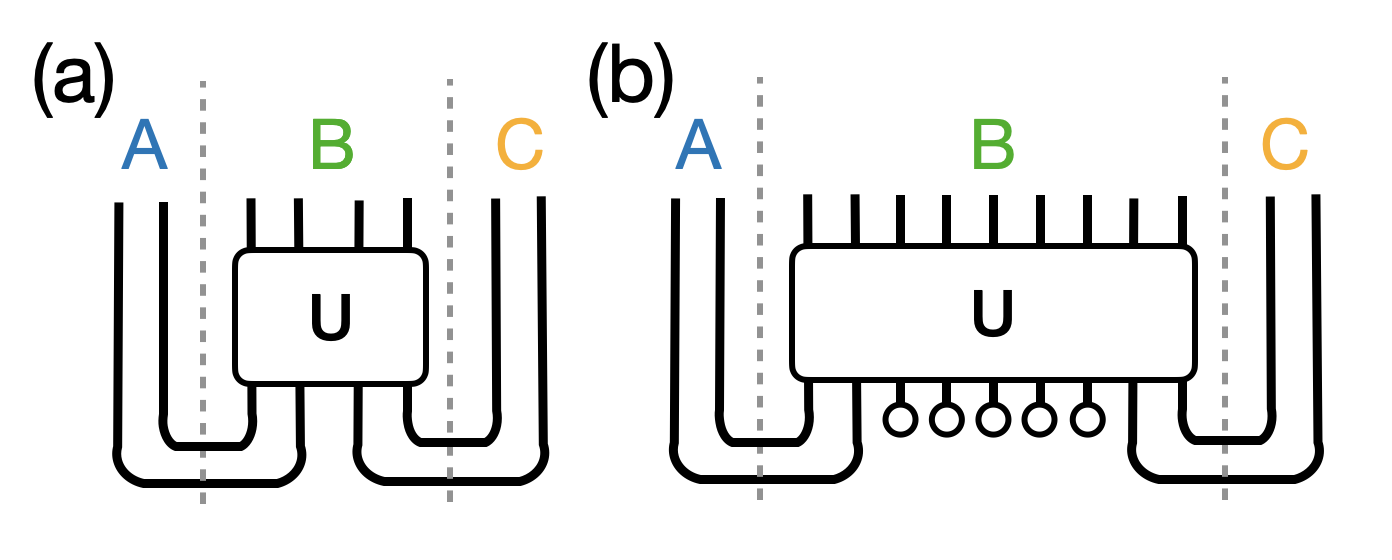}}
\caption{\label{fig:efficient_eraser}(a) An efficient teleporter where qubits in B are randomly erased. (b) When the size of B is bigger than AC, not all qubits can produce maximal CMI growth upon erasure.}
\end{figure}

% We can construct a class of efficient teleporters under erasure by the following procedure: pair each qubit in $AC$ with a qubit in $B$ and prepare each pair in a maximally entangled state, then apply a random unitary rotation on $AC$. In this case, any one of $\min(N_A, N_C)$ of the qubits in $B$ can be erased giving the maximum $\delta I$, in the limit where all subsystems have $N \to \infty$. 

% One might ask if it is possible to do better: e.g., if any one of the qubits in $N_B$ can be measured giving maximum growth in $\delta I$. This cannot happen by the following logic. Suppose there are two qubits $E_1$ and $E_2$ in $B$ that both saturate the bound on $\delta I$. Then $I(E_1:BE_2) = I(E_2:BE_1) = 0$. It follows that $S(B E_1 E_2) = S(B) + S(E_1) + S(E_2)$, so $E_1 E_2$ decouples from $B$. This argument can be iterated to show that if $k$ qubits in $B$ have the property that erasing any of them gives maximal CMI gain, then erasing all of them will also give maximal CMI gain. It follows that at most $\min(N_A, N_C)$ qubits have the property of causing maximal CMI gain upon erasure. 

\subsection{Efficient teleporters under measurement}

We now consider the case of efficient teleporters where qubits in B are randomly subject to measurement channels. In this case, a class of efficeint teleporters can be constructed by preparing multiple copies of Bell teleportation protocols. However, Bell teleportation also has the structure that \(N_B=N_A+N_C\) so we again ask if efficient teleportation is possible when the size of \(B\) is bigger than  the size of \(AC\).

Contrast to the case of erasure channel, we uncover a class of states where measuring \textit{any} subset of qubits in \(B\) produces maximal CMI growth (with probability exponentially close to 1), no matter how big the size of \(B\) is. This is inspired by the Hayden-Preskill Thought Experiment \cite{10.1088/1126-6708/2007/09/120} and its decoder due to Yoshida and Kitaev \cite{10.48550/arxiv.1710.03363}. In our setup, as depicted in Fig.\ref{fig:hayden_preskill}(a), systems A and B, as well as B and C, initially share maximally entangled states. Additionally, B possesses a larger maximally entangled state. Half of B's states undergo evolution under a scrambling, Clifford random unitary \(U\), while the other half evolve under its time-reversed counterpart \(U^*\). Subsequently, states in B are paired up and measured in the Bell basis. We randomly select pairs in B to measure and then compute \(\delta I\), followed by averaging over these measurement choices. The result is shown in Fig.\ref{fig:hayden_preskill}(b). Regardless of B's size, the averaged \(\delta I\) consistently reaches the upper bound, indicating that the specific qubits measured and B's size are irrelevant to efficient teleportation.

 Efficient teleportation of the Clifford Hayden-Preskill protocol has been analyzed in \cite{10.48550/arxiv.2106.15628} and we review it in Appendix \ref{clifford_hayden_preskill}. In alignment with our observations on maximal CMI growth under measurement channels, we offer a novel perspective on this protocol's efficiency irrespective of the qubits measured. Since \(\rho_{AC} \propto\) the identity matrix, Theorem \ref{thm:structural_theorem} implies that the Clifford Hayden-Preskill protocol is equivalent to Bell teleportation. This equivalence is evident when considering the final Bell measurements, as measuring the \(XX\) and \(ZZ\) stabilizers. In Fig. \ref{fig:hayden_preskill}(c), the left panel shows the stabilizer \(Z_{1L}Z_{1R}\) supported on qubits \(B_{1L},B_{1R}\). Commuting through the Clifford unitary \(U\), this stabilizer transforms to \(Z_{1L}(t)Z_{1R}(-t)\), where \(Z_{1L}(t) = U^\dagger Z_{1L} U\) and \(Z_{1R}(-t) = U^T Z_{1R} U^*=(U^\dag Z_{1R} U)^T\). Importantly, the Clifford unitary maps a Pauli operator to another Pauli operator, leading to \(Z_{1L}(t)\) and \(Z_{1R}(-t)\) factorizing into identical strings of Pauli operators, as visualized in the right panel of Fig. \ref{fig:hayden_preskill}(c). The large maximally entangled state in B are the eigenstates of the green Pauli operators, while the remaining blue and yellow Pauli operators become some Bell stabilizers on A and C. This analysis extends to \(XX\) stabilizers and other qubits as well, confirming that the Clifford Hayden-Preskill protocol adheres to Theorem \ref{thm:structural_theorem}.

The irrelevance of which qubits are measured in B emerges from the scrambling nature of \(U\). Generally, the transformed stabilizers \(Z_{1L}(t)Z_{1R}(-t)\) correspond to random products of Bell stabilizers in systems A and C. For example, \(Z_{1L}(t)Z_{1R}(-t)\) corresponds to \(X_{A1}X_{C1}Y_{A2}Y_{C2}Z_{A3}Z_{C3}...\), \(Z_{2L}(t)Z_{2R}(-t)\) corresponds to \(Z_{A1}Z_{C1}Y_{A3}Y_{C3}Z_{A7}Z_{C7}...\), etc.  When N pairs of qubits in B are randomly selected and their corresponding stabilizers measured, with high probability these stabilizers are independent as long as N is smaller than the sizes of A and C. Measuring these stabilizers results in a mixed stabilizer state in systems A and C, generated by these independent Bell stabilizers, thereby teleporting N bits of information. This structure underlies the reason why specific qubits chosen for measurement in B are irrelevant.

Crucially, the efficient teleportation of this protocol requires \(U\) to be a Clifford unitary. In contrast, for a generic Haar-random \(U\), CMI growth becomes significanty sub-maximal. We analyse the CMI growth of the Haar Hayden Preskill protocol in Appendix \ref{haar_hayden_preskill}. Here, we provide an intuitive reason on why the Haar Hayden Preskill protocol cannot saturate the CMI upper bound. Returning to Fig.\ref{fig:hayden_preskill}(c), if \(U\) is Haar-random, then \(Z_{1L}(t)Z_{1R}(-t)\) will not factorize into Pauli strings but have operator entanglement. Measuring it will then collapse the state into a subspace where the component on \(AC\) is entangled with the remaining degrees of freedom in \(B\), and thus become indistinguishable.

\begin{figure}
\scalebox{0.35}{\includegraphics{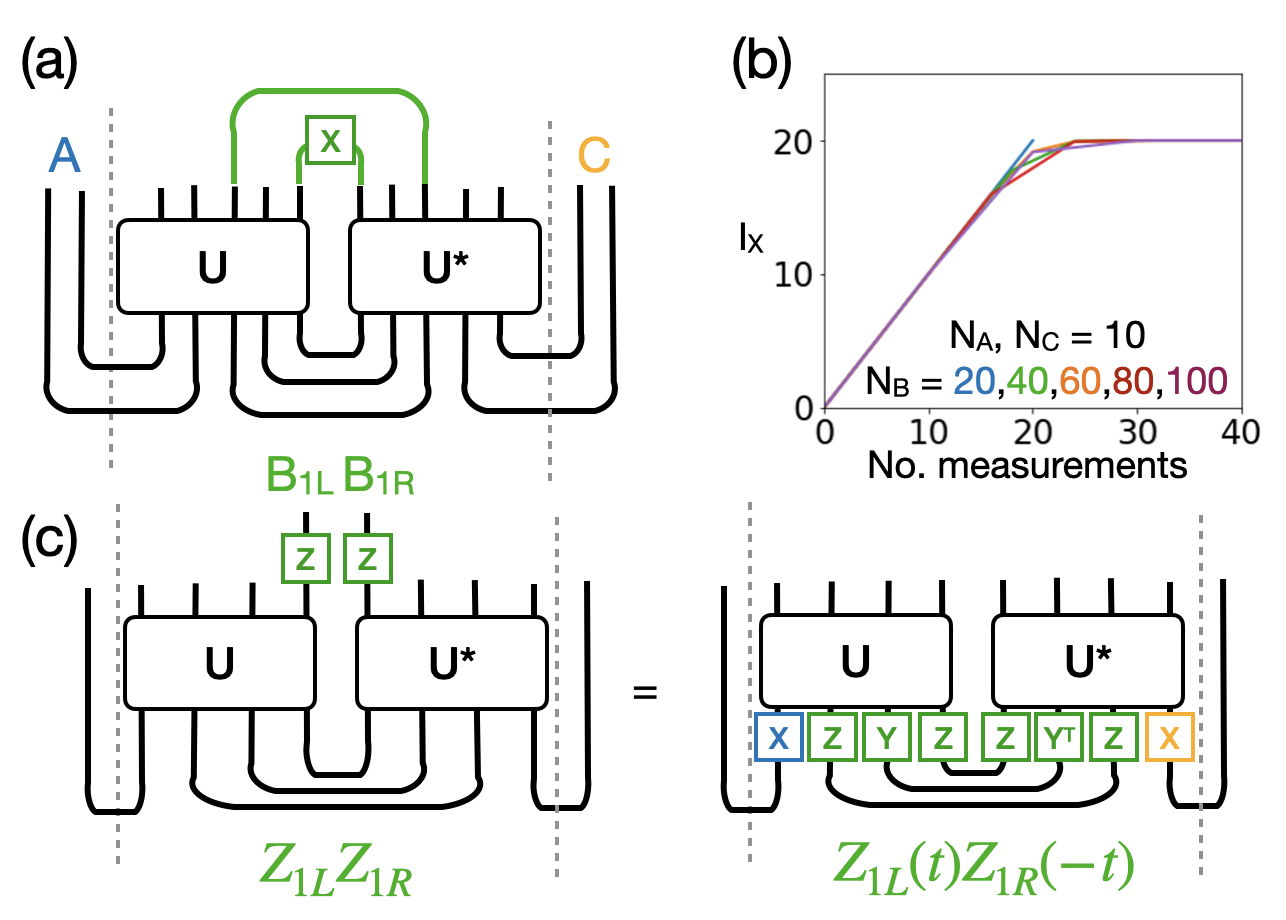}}
\caption{\label{fig:hayden_preskill}(a) Schematics of the Hayden-Preskill Protocol.  (b) MIE as a function of the number of Bell measurements. Each bell measurement is counted twice because it generates two bits of information. Irrespective of the size of B, the growth of MIE is always maximal until it reaches the maximal value. (c) Performing Bell measurements in \(B\) is equivalent to measuring some Bell stabilizers in \(AC\).}
\end{figure}

\section{CMI under Random Unitary Circuits}\label{ruc}

Finally, we turn to the evolution of CMI in more generic quantum states sampled from random unitary circuits, where the depth can be tuned to controlled the degree of non-local interactions. We will first discuss the case of measurements, as it is more intuitive, and makes contact with previous literature on the teleportation transition. Then we will discuss how replacing measurements with erasures leads to a qualitative change in the nature of this transition. 

Throughout this discussion we will consider a tripartite system $ABC$ with $N_A, N_B, N_C$ qubits in the respective subsystems. We will take $N_A = N_C$ for simplicity. We will assume that all three subsystems are extensive in the system size $N$; our primary concern is with the thermodynamic limit $N \to \infty$.

\subsection{Case of measurements}

To orient this discussion, we first consider the growth of CMI in Haar-random states. By Proposition \ref{mie_cmi}, the CMI after measuring $k$ qubits in $B$ is simply the Born-averaged mutual information $I(A:C|i)$ in the post-measurement state. The post-measurement state is a Haar-random state on $N - k$ qubits. In this state, the mutual information $I(A:C) = \max(0, N_A + N_C - (N_B - k))$ (up to Page corrections)~\cite{PhysRevLett.71.1291}. 

\vthree{A random circuit saturates the bipartite entanglement at depth $O(L)$  \cite{10.1007/s00220-016-2706-8,10.48550/arXiv.1905.12053}, so they behave similarly to Haar-random states in terms of bipartite entanglement.}
%
% Haar-random states---from the point of view of entanglement structure---are reached after random quantum circuits of depth $O(L)$ \cite{10.1007/s00220-016-2706-8,10.48550/arXiv.1905.12053}. 
This establishes that the CMI in these states begins to grow to a nonzero value when a finite fraction of the spins in $B$ have been measured. It is interesting to ask how this phase transition connects to the ``teleportation'' transition that takes place at depth $O(\log L)$ in one dimension and \(O(1)\) in higher dimensions \cite{10.48550/arXiv.2110.06963}. At the teleportation transition, the development of CMI between $A$ and $C$ requires \emph{all} the spins in $B$ to be measured.

To support this analysis, we explore the dynamics of CMI in random unitary Clifford circuits consisting of two-qubit gates laid out in a brickwork pattern, as depicted in Fig. \ref{fig:inefficient_teleporter}(a). While we use a 1+1D circuit for illustration, our findings are not dependent on the dimensionality of the circuit. We plot \(\delta I\) as a function of the circuit depth and the number of measurements in Fig.\ref{fig:inefficient_teleporter}(b). We see both regimes of behavior---the depth-independent Haar-random behavior at depth $O(L)$ (green line), and the onset of teleportation at depth $O(\log L)$ (blue star), and a transition at intermediate depths that interpolates linearly between these (red line). Additional numerics and scaling analysis are presented in Appendix \ref{scaling} to support our claim on the scaling behavior.

This interpolation can be readily understood within the statistical mechanics treatment of random circuits. We explain the intuition here and defer additional numerics to Appendix \ref{stat_mech}. In the statistical mechanics framework, the state after a circuit of depth $t$ can be regarded as a system of fictitious spins with ferromagnetic interactions of range $\sim t$. To compute the entanglement entropy of part of a pure state, one imposes distinct biasing fields to the region and its complement, and computes the free energy of the system subject to these boundary conditions. The qubits that have been measured have no bias field acting on them. In this statistical mechanics model, the teleportation transition can be understood as follows: region $B$ is fully measured so it is a region with no bias field. To compute the MIE, which is just $S(A)$ in the post-measurement state, one imposes opposite boundary conditions on $A$ and $C$, forcing a domain wall between them. When the circuit is shallow enough, region $B$ is in the paramagnetic phase and the domain wall costs no energy in the limit $N_B \to \infty$; beyond a critical depth (in $d \geq 2)$ or at logarithmic depth (in $d = 1$) region $B$ is ferromagnetically correlated, so the free energy cost of the domain wall remains finite as $N_B \to \infty$, and consequently the MIE is nonzero. 

When a fraction of the spins in $B$ remain unmeasured, they act as a net biasing field in $B$. The CMI is now a mutual information rather than an entanglement, so it is a difference of domain wall energies (Appendix \ref{stat_mech}). At finite depth, when $B$ is large enough, the measured spins in $B$ are polarized by the unmeasured spins, so there is no correlation between $A$ and $C$ when $B$ is large enough---as one might expect, there is no distinction between the paramagnet and the ferromagnet in a nonzero field. A teleportation transition does occur at depth $\sim N$, when the energy cost $\sim N$ of creating domain walls at the $AB$ and $BC$ interfaces exceeds the energy gain from polarizing the unmeasured spins in $B$ along the local field. This transition is discussed in more detail in App. \ref{stat_mech}.

% To support this analysis, we explore the dynamics of CMI in random unitary Clifford circuits, as depicted in Fig. \ref{fig:inefficient_teleporter}(a). The circuit features two-qubit gates arranged in a brickwork pattern, allowing us to control the degree of non-locality by adjusting the circuit depth. While we use a 1+1D circuit for illustration, our findings are not dependent on the dimensionality of the circuit. %In practice, we perform numerics using Clifford circuits which forms a 3-design []. 
% %
% We plot \(\delta I\) as a function of the circuit depth and the number of measurements in Fig.\ref{fig:inefficient_teleporter}(b). We see both regimes of behavior---the depth-independent Haar-random behavior at depth $O(L)$, and the onset of teleportation at depth $O(\log L)$, and a transition at intermediate depths that interpolates linearly between these. %The linear interpolation can be readily understood within a mean-field description of the statistical mechanics treatment of random circuits [***].

\subsection{Case of erasures}

We now turn our attention to applying erasure channels to states sampled from random unitary circuits. Again we begin with the case of a Haar-random state, corresponding to ciruits of  \(O(L)\) depth. 
%
%Finally, we comment on how these results generalize to the case of erasure errors instead of measurements. Again, we first consider the Haar-random state. 
The CMI after decoherence is that of a random state distributed between $ABC$ and the environment $E$, which contains $k$ qubits that were originally in $B$. It is straightforward to see that the CMI is nonzero when $|N_B - 2k| \leq N_A + N_C$. This behavior is richer than that in the case of measurements: too many erasures can cause the CMI to decrease after initially rising. 

We now turn to the situation at finite depth. From the definition of CMI, we note that $B$ and $E$ are always in complementary subsystems, so in the statistical mechanics mapping they always have opposite fields acting on them. Since the erasures are chosen randomly, the CMI calculation amounts to evaluating the free energy cost of a domain wall in a ferromagnet in the presence of random local fields. When $B$ and $E$ are different sizes, there is a net field, and (as discussed in the previous section) no teleportation transition occurs until depths linear in $N$. But when $B$ and $E$ are the same size, there is no net field, and a finite-depth transition from a random-field paramagnet to a ferromagnet can happen in $d \geq 3$ (but not in $d = 1, 2$ according to the Imry-Ma argument~\cite{PhysRevLett.35.1399}). This transition corresponds to the generation of long-range CMI, and is a generalization of the teleportation transition in the case of measurements.

Again we support our analysis with Clifford numeric shown in Fig. \ref{fig:inefficient_teleporter}(c). The circuit is arranged in the brickwork architecture with depth being tuned. Instead of measurements in the end, we randomly erase qubits in B and observe the behavior of CMI. We plot \(\delta I\) as a function of the circuit depth and the number of erasures in Fig.\ref{fig:inefficient_teleporter}(b). Again we observe the teleportation behavior consistent with the Haar behavior at \(O(L)\) circuit depth (green line). The predicted rise and fall in CMI appears in the window where $|N_B - 2k| \leq N_A + N_C$.  A teleportation transition happens at shallow depth (blue star) and interpolates to the Page behavior with increasing circuit depth (red line). 

In $d = 1, 2$ the correlation length remains finite at all finite depth. Indeed, following the Imry-Ma argument, we predict that in one dimension the correlation length over which information is teleported grows as $\xi(t) \sim t^2$ \footnote{we note that this phenomena was firstly discovered in [Lee, Jiang] in a slightly different setup. We thank the authors for useful discussions}, and correspondingly teleportation transition only happens at depth $L^{1/2}$, while it takes place at logarithmic depth in two dimensions and finite depth in three or more dimensions. This scaling is confirmed at Appendix \ref{scaling}.

%Notably, we observe a teleporting phase at deep circuits and high number of measurements. When the circuit depth is very low (blue dashed line), no teleportation is observed regardless of the number of measurements. However, beyond a critical depth (blue star), the circuit prepares a resource state, enabling the teleportation of \(O(1)\) bit of information by measuring the entirety of system B. This regime has been previously analyzed in [reference needed].

%As the circuit becomes sufficiently deep (green dashed line), the efficiency of teleportation becomes largely independent of the depth. Unlike the shallow circuit regime where the entire system B has to be measured, here measuring only a fraction of the qubits in B is sufficient to observe teleportation. This regime is sometimes referred to as deep thermalization [reference needed] and can be analytically interpreted using Page's theorem, since the circuit approaches a global Haar-random unitary [reference needed].

%In the regime of intermediate circuit depth (purple dashed line), we observe a trade-off between the circuit depth—representing the degree of non-locality—and the number of measurements required for teleportation. States with higher degree of non-locality necessitate fewer measurements to achieve teleportation. This phase diagram can be qualitatively understood through the statistical mechanics models [reference needed], a topic we delve into further in [reference needed].

\begin{figure}
\scalebox{0.4}{\includegraphics{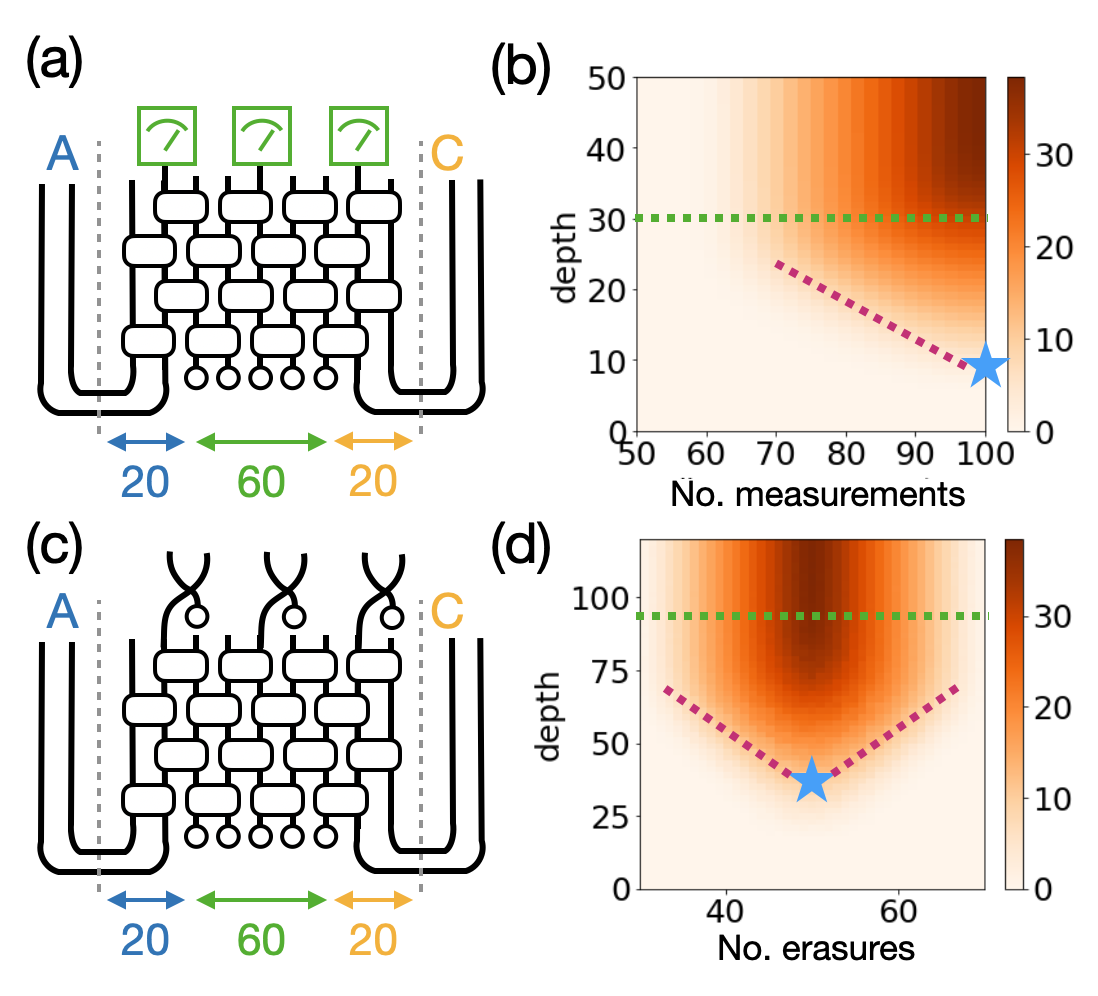}}
\caption{\label{fig:inefficient_teleporter}(a) States prepared with random unitary circuits are subject to measurement channels on randomly selected qubits. For the data shown here, \(N_A=N_C=20, N_B=100\).  (b) \(\delta I\) as a function of the number of measurements and circuit depth. (green line) Page behavior at \(O(L)\) circuit depth. (blue star) onset of teleportation transition. (red line) Interpolation from teleportation transition to Page behavior. (c) States prepared with random unitary circuits are subject to erasure channels on randomly selected qubits. For the data shown here, \(N_A=N_C=20, N_B=100\).  (b) \(\delta I\) as a function of the number of erasures and circuit depth. (green line) Page behavior at \(O(L)\) circuit depth. (blue star) onset of teleportation transition. (red line) Interpolation from teleportation transition to Page behavior.}
\end{figure}

\section{Conclusion}

To summarize, we have discussed the nonlocal dynamics of quantum conditional mutual information under decoherence in the conditioning system. 
%
%the local decoherent dynamics of CMI. while decohering A and C cannot increase CMI as expected from the loss of data, decohering B can either decrease or increase CMI. 
%
We derived an upper bound for the rate at which decoherence can induce long-range CMI, both for the case of measurements---in which case the CMI can be related to the measurement-induced entanglement---and for generic channels. 
We also provided a structural characterization of systems in which nonlocal CMI is efficiently generated, showing that the nonlocal CMI generation is inherently similar to teleportation. 
In addition, we showed that measuring particular subsets of B is not necessary for achieving efficient teleportation. We give an explicit example of Clifford Hayden-Preskill protocol where each measurement corresponds to measuring random combination of Bell stabilizers. Measuring N qubits, the corresponding Bell stabilizers are independent with high probability, thus teleporting N bits of information. 
Lastly, we discussed the dynamics of CMI under random unitary circuits, connecting the teleportation transition at finite depth with the Page transition that occurs after full scrambling. This discussion also led to the intriguing observation that the finite-depth teleportation transition is strongly modified when measurements are replaced by erasure channels. Our results suggest that it might be fruitful to explore the growth of CMI under decoherence in a broader family of physically relevant states, such as topological phases, cluster states, and quantum error correcting codes. 

\emph{Note added}.---While we were completing this work we became aware of the work of Ref.~\footnote{S. Lee and L. Jiang, in preparation.} exploring the effect of depolarization on CMI in random unitary circuits, which will appear soon. Our results are complementary and agree where they overlap.

\begin{acknowledgments}
We are grateful to Andrew Potter, Shengqi Sang, and Romain Vasseur for helpful discussions. This work was funded by NSF DMR-2103938.
\end{acknowledgments}

{% It is less understood how CMI evolves in time, however. If the three distinct systems undergo local coherent dynamics without any interactions, the information content of each system remains unchanged, so CMI should remain unchanged. When the local dynamics is decoherent, one would naively expect CMI to decrease due to the lost of infomation. This is not true, however. When systems A and C are subject to decoherence, CMI indeed decreases due to information loss, as illustrated in Fig. \ref{fig:cmi_setup}(a). Yet, intriguingly, decoherence in the conditioning system B presents a more complex behavior. Depending on the specific dynamics, decohering system B can result in either an increase or a decrease in CMI, as depicted in Fig. \ref{fig:cmi_setup}(b). Notably, applying a dephasing channel to system B connects to teleportation in such many-body state as shown in Fig. \ref{fig:cmi_setup}(c). This observation opens a novel avenue for exploring many-body quantum teleportation. As a non-local phenomenon, many-body teleportation serves as a potent method for unveiling global, hidden information in quantum systems, such as sign structures and long-range entanglements [references needed]. By interpreting many-body teleportation through the lens of CMI, we gain an additional information-theoretic perspective, which complements the many-body physics viewpoint.

}

% \section{Clifford Hayden-Preskill}
% Having understood efficient teleporters, we now address the question of which qubits in B to be measured for efficient teleportation. 
\begin{widetext}
\appendix

\section{Comparing Different Teleportation Metrics}\label{teleport_metric}
In this section, we discuss and compared different teleportation metrics commonly used in the literature. Teleportation in many-body states is often quantified by the concept of measurement-induced entanglement (MIE) \cite{10.22331/q-2023-02-02-910} and localizable entanglement \cite{10.1103/physreva.71.042306}. MIE is defined for a pure state \(\ket{\psi}_{ABC}\) where all qubits in B are measured in the computational basis, without loss of generality. for each measurement outcome \(i\) with probability \(p_i\), there is a resulting pure state \(\ket{\psi}_{AC|i}\). MIE is defined as:
\begin{equation}
    MIE=\sum_i p_i S(A|i)
\end{equation}
Where \(S(A|i)\) is the entanglement entropy of \(\ket{\psi}_{AC|i}\). Localizable entanglement is similar to MIE, except one optimizes the measurement basis of each qubit in B to maximize MIE. 

To compare with our CMI definition, We note that \(\delta I = 2 \; MIE\) in the scenario where all qubits in B are measured, but \(\delta I\) is also defined when B is partially measured or subject to a different decoherent channel. In addition, MIE is not well-defined when the initial state is mixed. 
% Therefore, CMI as a teleportation metric generalizes MIE to the case of arbitrary channels as well as the case of mixed states.
%
\vthree{Therefore, CMI as a non-local correlation measure generalizes MIE to the case of arbitrary channels as well as the case of mixed states.}

\section{Proof of Proposition \ref{thm:structural_theorem}}\label{proof}
In this section, we present the proof of Proposition \ref{thm:structural_theorem}. We begin by introducing a key lemma which states that by mixing some density matrices \(\rho_1,\rho_2,...,\rho_N\) at some probability, the resulting density matrix has upper-bounded entropy saturated iff different \(\rho_i\) are supported on orthogonal subspaces.

\begin{lemma}
   \label{lemma:ortho_lemma} Let  \(\rho_1,\rho_2,...,\rho_N\) be a collection of density matrices. Let \(\rho=\sum_{i=1}^{N} p_i \rho_i\) be a mixture of all density matrices with probability p, then,
\begin{equation}
\label{eq:ortho_lemma}
    S(\rho) \le \sum_{i=1}^{N} p_i S(\rho_i) + S' 
\end{equation}
Where \(S'=-\sum_{i=1}^N p_i \; log(p_i)\) denotes the classical entropy associated with the mixing probability. Moreover, the equality is taken iff \(\rho_i \rho_j=0,\forall i\neq j\).
\end{lemma}
The proof is technical, so we will first present the intuition behind this lemma and its implication. Mixing quantum states obviously increases their entropy. The resulting entropy contains two components: the first component comes from the entropies of the initial states to be mixed, shown as the the first term in the right-hand side of (\ref{eq:ortho_lemma}); the second component is due to mixing itself, shown as the the second term in the right-hand side of (\ref{eq:ortho_lemma}). The first part of the theorem states that the final entropy is upper-bounded by the sum of the two component. The second part of the theorem characterizes the situation when the upper-bound is saturated: one has to mix orthogonal states to maximize the final entropy. 
One can also comprehend this lemma from the the perspective of quantum Shannon theory. By re-organizing the terms in the inequality, we obtain \(S(\rho) - \sum_{i=1}^{N} p_i S(\rho_i) \le + S' \). The left-hand side is  known as the Holevo information \cite{nielsen_chuang} which characterizes the capacity of sending classical information with quantum states. It is known that Holevo information is upper-bounded by \(S'\) and it is also known that sending orthogonal states as messages saturates this bound. Non-trivially, we further prove that sending orthogonal states as messages is the only way to saturate the bound.

\begin{proof}
    We begin by mixing two density matrices \(\rho_1\) and \(\rho_2\), and then generalize to mixing N density matrices. Let \(\rho=p \rho_1 + (1-p)\rho_2\). We evaluate \(S(\rho)\) explicitly.
\begin{align}
       S(\rho)&=-Tr[(p \rho_1 + (1-p) \rho_2)  \log(p \rho_1 + (1-p) \rho_2)] \\
       \label{eq:ortho_lemma_expansion_1} 
       &=-Tr[p \rho_1  \log(p)]-Tr[p \rho_1  \log(\rho_1 + \frac{1-p}{p}\rho_2)]\\
       \label{eq:ortho_lemma_expansion_2} 
       &-Tr[(1-p) \rho_2  \log (1-p)]-Tr[(1-p) \rho_2  \log(\rho_2 + \frac{p}{1-p}\rho_1)]
\end{align}
The first and third term are equal to \(S'\), so
\begin{equation}
     S(\rho)=S'-Tr[p \rho_1  \log(\rho_1 + k_1\rho_2)]-Tr[(1-p) \rho_2  \log(\rho_2 + k_2\rho_1)]
\end{equation}
Where we set \(k_1=\frac{1-p}{p} \ge 0, k_2=\frac{p}{1-p} \ge 0\) to simplify the notation. We would like to upper-bound the two trace terms. Focusing on the first term, we observe that \(\rho_1+k_1\rho_2 \ge \rho_1\) in the operator sense: \(A \ge B\) means that \(A-B\) is a positive semidefinite matrix, which is true here because \(\rho_1+k_1 \rho_2 - \rho_1=k_1 \rho_2\) is indeed positive semidefinite. Next, the matrix log is a operator monotone \cite{10.7153/jmi-07-08} so that \(\log(\rho_1+k_1\rho_2) \ge \log(\rho_1)\) as well. Therefore, we can upper-bound the first trace term by replacing \(\log(\rho_1+k_1\rho_2)\) with \(\log(\rho_1)\).
\begin{align}
    -Tr[p \rho_1  \log(\rho_1 + k_1\rho_2)] &= -Tr[p \rho_1  \log(\rho_1)] - Tr[p \rho_1  (\log(\rho_1 + k_1\rho_2)-\log(\rho_1))] \\
    \label{eq:trace_bound_1}&\le -Tr[p \rho_1  \log(\rho_1)] = p S(\rho_1)
\end{align}
Where the last inequality is because \(p \rho_1\) and \(\log(\rho_1 + k_1\rho_2)-\log(\rho_1)\) are both Hermitian and positive semidefinite, so their product is also Hermitian and positive semidefinite, hence the second trace is \(\ge 0\). We can obtain similar results for the second trace term in (\ref{eq:ortho_lemma_expansion_2}), namely
\begin{align}
    -Tr[(1-p) \rho_2  \log(\rho_2 + k_2\rho_1)] &= -Tr[(1-p) \rho_2  \log(\rho_2)] - Tr[(1-p) \rho_2  (\log(\rho_2 + k_2\rho_1)-\log(\rho_2))] \\
    \label{eq:trace_bound_2}&\le -Tr[(1-p) \rho_2  \log(\rho_2)] = (1-p) S(\rho_2)
\end{align}
Plugging the upper bound (\ref{eq:trace_bound_1}) and (\ref{eq:trace_bound_2}) into (\ref{eq:ortho_lemma_expansion_1}) and (\ref{eq:ortho_lemma_expansion_2}), we obtain
\begin{equation}
    S(\rho) \le p S(\rho_1) + (1-p) S(\rho_2) + S'
\end{equation}
Where the inequality is saturated iff \(\log(\rho_1 + k_1\rho_2)-\log(\rho_1)=0\) and \(\log(\rho_2 + k_2\rho_1)-\log(\rho_2)=0\). This completes the proof of the first part of the lemma.

To prove the second part of the the lemma which relates the saturation of the inequality to the support of \(\rho_1\) and \(\rho_2\), we utilize the integral representation of the matrix logarithm \cite{10.7153/jmi-07-08}.
\begin{equation}
    \log(\rho)=\int_0^{\infty} (\frac{1}{1+t}I-(\rho + t I)^{-1}) dt
\end{equation}
Where \(I\) is the identity matrix and \(t\) is an auxiliary variable we integrate over. In the integral representation, \(\log(\rho_1 + k_1\rho_2)-\log(\rho_1)\) can be written as
\begin{equation}
    \log(\rho_1 + k_1\rho_2)-\log(\rho_1)=\int_0^{\infty} ((\rho_1+t I)^{-1}-(\rho_1 + k_1 \rho_2 + t I)^{-1}) dt
\end{equation}
Importantly, if \(A \ge B\), then \(B^{-1} \ge A^{-1}\). Since \(\rho_1 + k_1 \rho_2 + t I \ge \rho_1+t I\), this would imply \((\rho_1+t I)^{-1}-(\rho_1 + k_1 \rho_2 + t I)^{-1}\) being positive semidefinite \(\forall t\). Therefore, to have \(\log(\rho_1 + k_1\rho_2)-\log(\rho_1)=0\), we need \((\rho_1+t I)^{-1}-(\rho_1 + k_1 \rho_2 + t I)^{-1}=0, \forall t\).

Since we have converted the messy matrix log function into some nicer expression, we can evaluate it directly. To understand the support of \(\rho_1\) and \(\rho_2\), we rewrite them in the following way:
\begin{align}
    \rho_1 &= \begin{pmatrix}
\widetilde{\rho} & 0 \\
0 & 0 
\end{pmatrix} \\
\rho_2 &= \begin{pmatrix}
\widetilde{\rho_{11}} & \widetilde{\rho_{12}} \\
\widetilde{\rho_{12}}^\dag & \widetilde{\rho_{22}} 
\end{pmatrix}
\end{align}
Where \(\rho_1\) is supported on the upper-left block with value \(\widetilde{\rho}\), and we divide \(\rho_2\) into four components: the component in the same subspace as \(\rho_1\) which we call \(\widetilde{\rho_{11}}\) , the component in the orthogonal subspace which we call \(\widetilde{\rho_{22}}\) , and the coupling between two subspaces which we call \(\widetilde{\rho_{12}}\). Having \(\rho_1\) and \(\rho_2\) supported on orthogonal subspaces is equivalent to having  \(\widetilde{\rho_{11}}=0\)  and  \(\widetilde{\rho_{12}}=0\). To show that, we evaluate the condition \((\rho_1+t I)^{-1}-(\rho_1 + k_1 \rho_2 + t I)^{-1}=0, \forall t\) in this basis. To begin with,
\begin{align}
    \rho_1+t I &= \begin{pmatrix}
\widetilde{\rho} + t I & 0 \\
0 & t I
\end{pmatrix} \\
\rho_1 + k_1 \rho_2 + t I &= \begin{pmatrix}
\widetilde{\rho} + k_1 \widetilde{\rho_{11}} + tI & k_1 \widetilde{\rho_{12}} \\
k_1\widetilde{\rho_{12}}^\dag & k_1\widetilde{\rho_{22}} +tI
\end{pmatrix}
\end{align}
Now we calculate the matrix inverse using Schur's complement. We only calculate the upper-left block since that is all we need.
\begin{align}
    (\rho_1+t I)^{-1} &= \begin{pmatrix}
(\widetilde{\rho} + t I)^{-1} & 0 \\
0 & \frac{1}{t} I
\end{pmatrix} \\
(\rho_1 + k_1 \rho_2 + t I)^{-1} &= 
\begin{pmatrix}
(\widetilde{\rho} + k_1 \widetilde{\rho_{11}} + tI - \widetilde{\rho_{12}}(\widetilde{\rho_{22}}+ tI)^{-1}\widetilde{\rho_{12}}^\dag )^{-1} & ... \\
... & ...
\end{pmatrix}
\end{align}
In order to have  \((\rho_1+t I)^{-1}-(\rho_1 + k_1 \rho_2 + t I)^{-1}=0\), we need the upper-left block to cancel out. Thus,
\begin{equation}
    \rho_1+t I = \widetilde{\rho} + k_1 \widetilde{\rho_{11}} + tI - \widetilde{\rho_{12}}(\widetilde{\rho_{22}}+ tI)^{-1}\widetilde{\rho_{12}}^\dag, \forall t
\end{equation}
The only way this relation is satisfied is to have  \(\widetilde{\rho_{11}}=0\) and \(\widetilde{\rho_{12}}=0\), which is exactly the orthogonality condition. To sum up, we have proven that if  (\ref{eq:ortho_lemma}) is satisfied for mixing two matrices, then they have to be supported on orthognal subspaces.

Generalizing to mixing N states can be accomplished by mixing sequentially. First we mix \(\rho_1\) and \(\rho_2\) with probability \(\frac{p_1}{p_1+p_2}\) and \(\frac{p_2}{p_1+p_2}\) to obtain \(\rho_1'=\frac{p_1}{p_1+p_2} \rho_1 + \frac{p_2}{p_1+p_2} \rho_2\). Next, we mix \(\rho_1'\) with \(\rho_3\) with probability \(\frac{p_1+p_2}{p_1+p_2+p_3}\) and \(\frac{p_3}{p_1+p_2+p_3}\) to obtain \(\rho_2'=\frac{p_1+p_2}{p_1+p_2+p_3} \rho_1' + \frac{p_3}{p_1+p_2+p_3} \rho_3\), and so on. Repeating this process, we obtain the final mixed state \(\rho=\sum_{i=1}^N p_i \rho_i\). At each step, we can apply the result for mixing two density matrices to upper-bound the entropy and to characterize the orthogonality between different \(\rho_i\).  This would produce the lemma for mixing N states.
\end{proof}
With Lemma \ref{lemma:ortho_lemma}, we are ready to prove the structural theorem of efficient teleporters. We restate the theorem below.
\begin{theorem}
    Consider \(\rho_{ABC}\) with \(I(A:C)=0\). Suppose \(B_1,B_2,...,B_N\) are measured in the computational basis. If \(\delta I=N\), then the resulting density matrix supported on \(A,B_1,B_2,...,B_N,C\) , namely \(\rho_{A\mathcal{N}[B_1B_2...B_N]C}\) has the following structure.
\begin{equation}
\rho_{A\mathcal{N}[B_1B_2...B_N]C}=\frac{1}{2^N}\sum_{i=0}^{2^N-1}\ket{i}\bra{i}_{B_1B_2...B_N} \rho_{AC|i}
\end{equation}
    Where the post-selected states \(\rho_{AC|i}\) are mutually orthogonal, namely \(\rho_{AC|i}\rho_{AC|j}=0\), \(\forall i \neq j\).
    \end{theorem}
\begin{proof}
    The probability \(\frac{1}{2^N}\) of each term comes from saturating the upper bound of teleportation \(\delta I=S(ABC)=N\). In order to have \(S(ABC)=N\), each measurement outcome has to be equally probably with probability \(\frac{1}{2^N}\).

We will use Lemma \ref{lemma:ortho_lemma} to prove mutual orthogonality. First, we rewrite \(\delta I\) in the following way.
\begin{equation}
\begin{split}
    \delta I &= \sum_i p_i I(A:C|i) \\
    &= \frac{1}{2^N}\sum_i^{2^N-1} S(A|i) + \frac{1}{2^N}\sum_i^{2^N-1} S(C|i) - \frac{1}{2^N}\sum_i^{2^N-1} S(AC|i)
\end{split}
\end{equation}
By convexity of the von Neumann entropy, \(\frac{1}{2^N}\sum_i^{2^N-1} S(A|i) \le S(A)\) and \(\frac{1}{2^N}\sum_i^{2^N-1} S(C|i) \le S(C)\). On the other hand, the Holevo bound (first part of Lemma \ref{lemma:ortho_lemma}) says that \(\frac{1}{2^N}\sum_i^{2^N-1} S(AC|i) \ge S(AC) - N\). Therefore, the only way to have \(\delta I=N\) is to saturate all three inequalities. In particular, since the last Holevo bound is saturated, we can invoke Lemma \ref{lemma:ortho_lemma} which immediately results in \(\rho_{AC|i}\rho_{AC|j}=0,\forall i \neq j\).
\end{proof}

With some additional assumptions, we can obtain the second part of Theorem 2.
\begin{corollary}
\label{col:mixed_bell_measurement} Given an efficient teleporter \(\rho_{ABC}\) teleporting N bits with \(\rho_{AC} = \frac{1}{2^{N_A+N_C}}I_{2^{N_A+N_C}}\), where \(N_A, N_C\) are the number of qubits in system A and C, and \(I_{2^{N_A+N_C}}\) is the identity matrix of size \(2^{N_A+N_C}\). Then \(\forall i\), \(\rho_{AC|i}\) has the following form in some basis:
    \begin{align}
\rho_{AC|i}=\begin{pmatrix}
0 & ... & 0 & ... & 0 \\
... & ... & 0 & ... & ... \\
0 & 0 & \widetilde{\rho_{AC|i}} & 0 & 0 \\
... & ... & 0 & ... & ... \\
0 & ... & 0 & ... & 0 
\end{pmatrix}&=\frac{1}{d}\begin{pmatrix}
0 & ... & 0 & ... & 0 \\
... & ... & 0 & ... & ... \\
0 & 0 & I_{d} & 0 & 0 \\
... & ... & 0 & ... & ... \\
0 & ... & 0 & ... & 0 
\end{pmatrix} \\
\rho_{A|i}&=\frac{1}{2^{N_A}} I_{2^{N_A}} \\
\rho_{C|i}&=\frac{1}{2^{N_C}} I_{2^{N_C}}
    \end{align}
    Where \(d=\frac{1}{2^{N_A+N_C-N}}\), and \(I_{d}\) is the identity matrix of size \(d\). Equivalently, \(\exists\) local unitaries \(Q_A, Q_C\) such that after application, the collection of post-selected states \(\{ (Q_A \otimes Q_C)\rho_{AC|i} (Q_A \otimes Q_C)^\dag\}_i\) are mix stabilizer states stabilized by some Bell stabilizers, and different \(i\) correspond to different signs of the Bell stabilizers.
\end{corollary}
\begin{proof}
We already know that different \(\rho_{AC|i}\) are supported on orthogonal subspaces, so one can find a basis such that
\begin{equation}
    \rho_{AC}=\frac{1}{2^N}\begin{pmatrix}
\widetilde{\rho_{AC|0}} & 0 & 0 & ... & 0 \\
0 & \widetilde{\rho_{AC|1}} & 0 & ... & 0 \\
0 & 0 & \widetilde{\rho_{AC|2}} & ... & 0 \\
... & ... & ... & ... & ... \\
0 & 0 & 0 & ... & \widetilde{\rho_{AC|2^N}}
\end{pmatrix}
\end{equation}
And under the assumption \(\rho_{AC} = \frac{1}{2^{N_A+N_C}}I_{2^{N_A+N_C}}\), individual \(\widetilde{\rho_{AC|i}}\) has to be proportional to the identity matrix as well. Therefore, we only have to show that
\begin{enumerate}
    \item \(\widetilde{\rho_{AC|i}}\) all have dimension \(d\).
    \item \(\rho_{A|i}=\frac{1}{2^{N_A}} I_{2^{N_A}}\) and \(\rho_{C|i}=\frac{1}{2^{N_C}} I_{2^{N_C}}\)
\end{enumerate}
Let us denote the dimension of \(\widetilde{\rho_{AC|i}}\) to be \(d_i\). Given that \(I(A:C|i)=S(A|i) + S(C|i) - S(AC|i)\), we can evaluate \(S(AC|i)\) directly and upper-bound \(S(A|i)\) and \(S(C|i)\).
\begin{align}
    S(A|i) &\le N_A \\
    S(C|i) &\le N_C \\
    S(AC|i)&=\log_2(d_i)
\end{align}
Therefore, \(I(A:C|i) \le N_A + N_C - \log_2(d_i)\). Averaging over all the outcome, we have
\begin{equation}
    \delta I = \frac{1}{2^N} \sum_{i=0}^{2^N-1} S(A|i) + \frac{1}{2^N} \sum_{i=0}^{2^N-1} S(C|i) - \frac{1}{2^N} \sum_{i=0}^{2^N-1} \log_2(d_i) \le  N_A + N_C - (N_A + N_C - N)
\end{equation}

To saturate the CMI upper bound, all inequalities have to be saturated, therefore we have
\begin{align}
    S(A|i) &= N_A, \forall i \\
    S(C|i) &= N_C, \forall i \\
    \frac{1}{2^N} \sum_{i=0}^{2^N-1} \log_2(d_i) &= N_A + N_C - N
\end{align}
The first two equality implies that \(\rho_{A|i}=\frac{1}{2^{N_A}} I_{2^{N_A}}, \rho_{C|i}=\frac{1}{2^{N_C}} I_{2^{N_C}}, \forall i\). Under the constrain \(\sum_{i=0}^{2^N-1} d_i = 2^{N_A+N_C}\), the only way to satisfy the last equality is to have \(d_i = d, \forall i\). This completes the proof of the first statement.

The second equivalent statement can be proven by explicit reconstruction. Randomly draw N Bell stabilizers. There are \(2^N\) mixed stabilizer states with the Bell stabilizers but different signs. One can always construct a unitary that rotates the mixed stabilizer states to  \(\rho_{AC|i}\) since they have the same entanglement structure.
\end{proof}
Note that the strategy to prove Corollary \ref{col:mixed_bell_measurement} does not work without the assumption \(\rho_{AC} = \frac{1}{2^{N_A+N_C}}I_{2^{N_A+N_C}}\). This is because post-selected states do not necessarily have the same amount of entanglement. For example, there is no upper bound to the entanglement of the post-selected states, even though their averaged value is upper-bounded.

\section{Teleportation in Hayden-Preskill Protocol}
In this section, we further elaborate on teleportation in the Hayden-Preskill protocol. We will briefly review the literature in the context of understanding the teleporting properties. Then we will discuss the Haar-random Hayden-Preskill protocol, why it is an inefficient teleporter, and some of its features.

\subsection{Yoshida-Kitaev Decoder}
We briefly review the Yoshida-Kitaev Decoder for the Hayden-Preskill protocol in this section without refering to much of its background \cite{10.48550/arxiv.1710.03363}. The key idea is that the maximally-entangled, fully-teleported state shown in the right-hand side Fig. \ref{fig:hayden_preskill_decoder} is the eigenstate of \(U \otimes U^*\). Therefore, \(U \otimes U^*\) becomes transparent when getting close to this eigen-subspace. Utilizing this feature, One could rotate to this eigenstate by projecting some qubits into the maximally entangled state to enhance the overlap with the fully-teleported state, as shown in the right-hand side  Fig. \ref{fig:hayden_preskill_decoder}. Crucially, Yoshida and Kitaev showed that if \(U\) is scrambling enough (in terms of exponentially small out-of-time correlator), then one only needs to project slightly more than \(N_A\) or \(N_C\) qubits. This is the rationale behind doing Bell measurements in the main text. 

When the Bell measurement yields other outcome other than \(\frac{1}{\sqrt{2}}(\ket{00}+\ket{11})\), in general there is no guarantee on teleportation, and the measurement protocol succeed with exponentially small probability. Nevertheless, when \(U\) is Clifford, then it does not matter what is the outcome of the Bell measurement because there is a local decoding procedure to convert the other three Bell states to \(\frac{1}{\sqrt{2}}(\ket{00}+\ket{11})\). We explain the decoder in the next section.
\begin{figure}
\scalebox{0.5}{\includegraphics{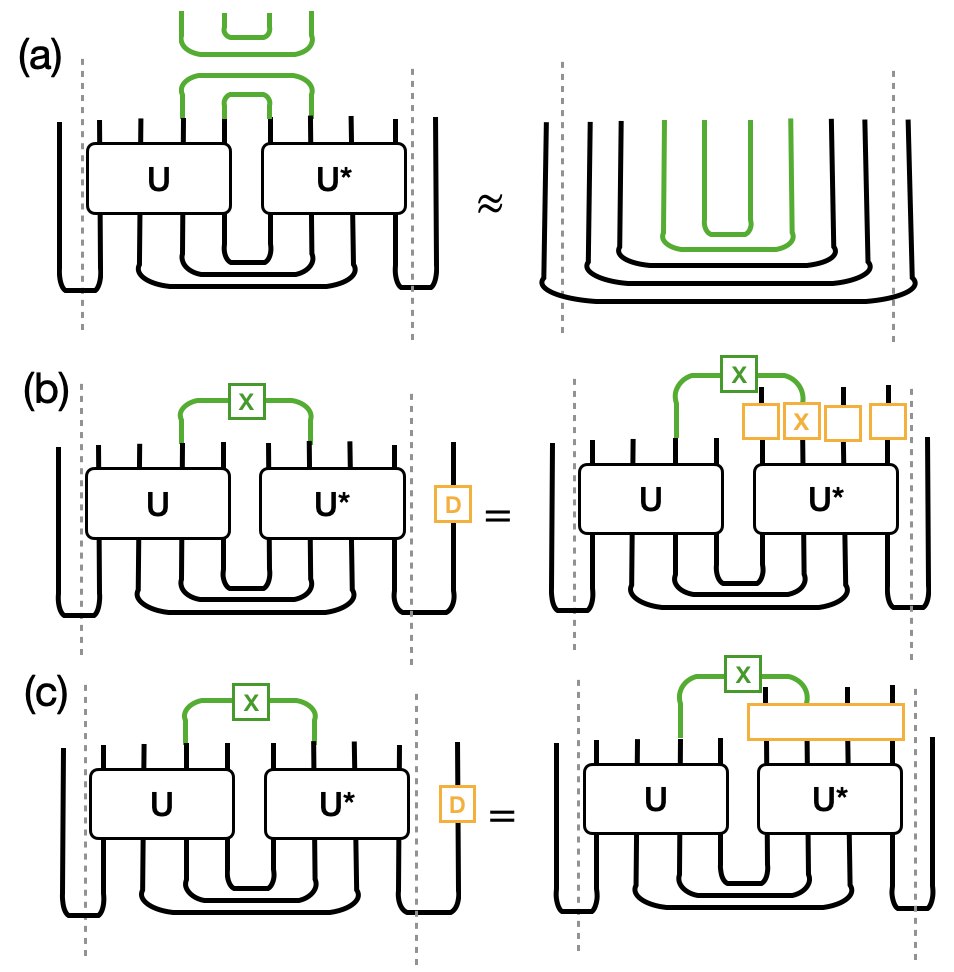}}
\caption{\label{fig:hayden_preskill_decoder}(a) projecting a few pair of qubits to \(\frac{1}{\sqrt{2}}(\ket{00}+\ket{11})\) rotate the state towards the maximally entangled eigenstate. (b) For Clifford Hayden Preskill, one can apply a decoding unitary to covert any other Bell basis back to \(\frac{1}{\sqrt{2}}(\ket{00}+\ket{11})\) (c) Similar decoding strategy is impossible in Haar Hayden Preskill because the decoding operation has operator entanglement.}
\end{figure}

\subsection{Efficient Teleportation in Clifford Hayden-Preskill via Decoding}\label{clifford_hayden_preskill}
We now explain why teleportation is efficient no matter what is the measurement outcome, given that \(U\) is Clifford. This section is adapted from \cite{10.48550/arxiv.2106.15628}. The three measurement outcomes other than \(\frac{1}{\sqrt{2}}(\ket{00}+\ket{11})\) correspond to having a Pauli operator insertion as shown in Fig. \ref{fig:hayden_preskill_decoder}(b). To remove the Pauli operator, we insert another Pauli operator \(D\) on the right-hand side. Since a Clifford operator maps a Pauli string to another Pauli string, we can commute \(D\) through \(U\) and replace it with the evolved Pauli string. Importantly, there is no operator entanglement in a Pauli string, so we can ignore most of the Pauli operators that do not intersect with the green line. For the Pauli operators that do intersect with the green line, we can always choose \(D\) such that after commuting through it evolves into the same Pauli operator as the ones coming from Bell measurement. Thus, they will be annihilated and result in the desired Bell projector \(\frac{1}{\sqrt{2}}(\ket{00}+\ket{11})\). 

On the other hand, a generic \(U\) will evolve \(D\) to some big operator with operator entanglement (Fig. \ref{fig:hayden_preskill_decoder}(c)). Therefore, this decoding method is not applicable beyond Clifford dynamics, and we will show in the next section that teleportation is indeed inefficient when \(U\) is Haar-random.

\subsection{Inefficient Teleportation in Haar-random Hayden-Preskill}\label{haar_hayden_preskill}
We now discuss inefficient teleportation in Haar-random Hayden-Preskill protocol and some of the features. We perform the same randomized simulation as described in the main text but choose \(U\) Haar-randomly. Fig. \ref{fig:haar_hayden_preskill}(a) shows how the averaged \(\delta I\) grows with increased number of measurements. Clearly teleportation is inefficient and becomes more inefficient when the size of B increases. In fact, one can observe a critical value above which \(\delta I\) grows linearly. This critical value appears to be at the point when the number of unmeasured qubits in B is roughly the size of A and C, regardless of what is the size of B. This behavior is consistent with Page's theorem, even though the underlying dynamics possess more structures.

To understand the special role of the Bell projector \(\frac{1}{\sqrt{2}}(\ket{00}+\ket{11})\) which we will call the good projector, we conduct another sequence of simulations where we first post-select the good projector for a few measurements, and then post-select one of the three other projectors (called bad projector) to fail the protocol. We increase the number of measurements where we post-select the good projector and compare the value of \(\delta I\) in Fig. \ref{fig:haar_hayden_preskill}(b,c). One can observe that \(\delta I\) quickly increases as long as we post-select the good projector. However, whenever a bad projector is selected, \(\delta I\) immediately collapses to the average value. We also observe the similar behavior in the fidelity to the eigenstate presented in the right-hand side of Fig. \ref{fig:hayden_preskill_decoder}(a). The fidelity at an increasing rounds of post-selection is shown in Fig. \ref{fig:haar_hayden_preskill}(c). As we post-select on the good projector for more measurements, the resulting state approaches the desired eigenstate. When a bad projector is post-selected, however, the fidelity immediately drops to zero. As an comparison, we compare the fidelity lower bound calculated in \cite{10.48550/arxiv.1710.03363} to our numerics and observe very tight compliance. 
\begin{figure}
\scalebox{0.5}{\includegraphics{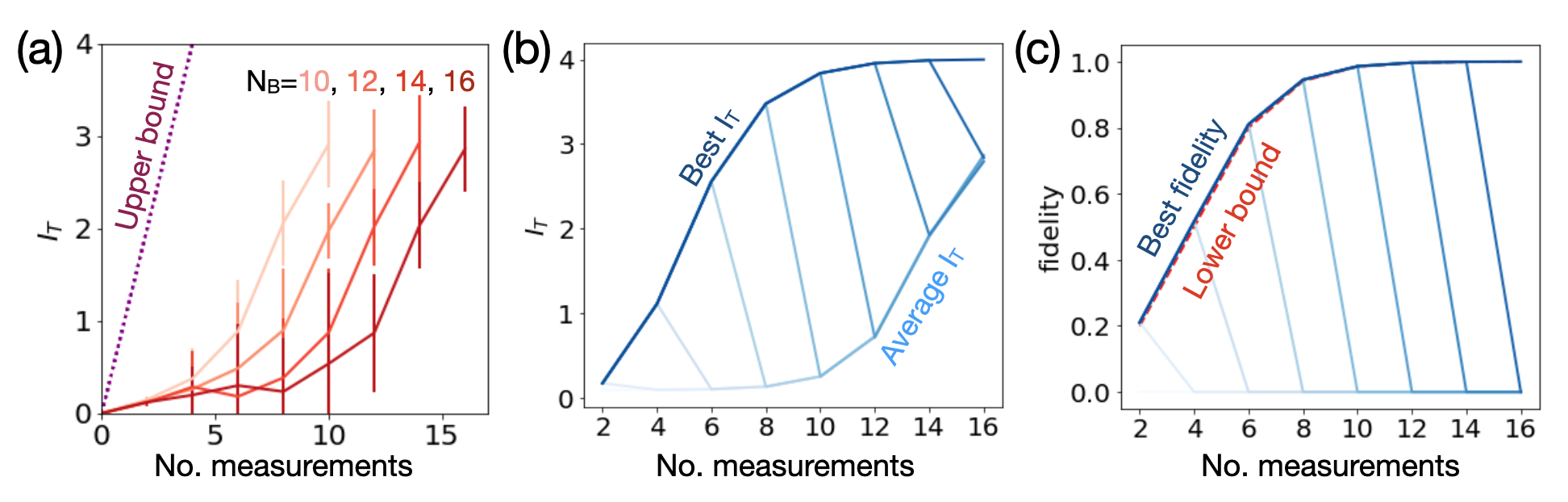}}
\caption{\label{fig:haar_hayden_preskill}(a) sub-maximal CMI growth in Haar Hayden Preskill. (b) Growth of CMI is close to maximal when one post-select the measurement outcomes to \(\frac{1}{\sqrt{2}}(\ket{00}+\ket{11})\), but post-selecting to any other Bell state immediately collapse the CMI to its average value (c) The fidelity to the maximally entangled eigenstate when one post-select the measurement outcomes to \(\frac{1}{\sqrt{2}} (\ket{00}+\ket{11})\). (red line) The fidelity lower bound calculated in \cite{10.48550/arxiv.1710.03363}.}
\end{figure}

\section{Teleportation in Random Unitary Circuits}
In this section, we further elaborate on the teleportation properties of the random unitary circuits. We begin by understanding the teleportation behavior as the circuit approaches the thermalization regime. Then, we qualitatively understand the shallow-to-intermediate depth regime via the statistical mechanics model.

\subsection{Transition into Thermalization Regime}
In this section, we focus on the regime when the circuit depth is high. We take the constant-depth cut shown in Fig. \ref{fig:thermalization}(a) and plot \(\delta I\) along the cut in Fig. \ref{fig:thermalization}(b). \(\delta I\) increases at higher depth until some point when the behavior of \(\delta I\) saturates (green line in Fig. \ref{fig:thermalization}(a)). The saturated curve comply with the Page behavior: one would only observe non-trivial \(\delta I\) when the number of unmeasured qubits in B is smaller than the size of A plus C, and \(\delta I\) grows linearly from there. To further test the consistency with the Page behavior, we replot the data in the log scale as shown in Fig. \ref{fig:thermalization}(c). \(\delta I\) exhibits a exponential decay below the critical number of measurements, which is also consistent with the exponentially small Page correction. 
\begin{figure}
\scalebox{0.5}{\includegraphics{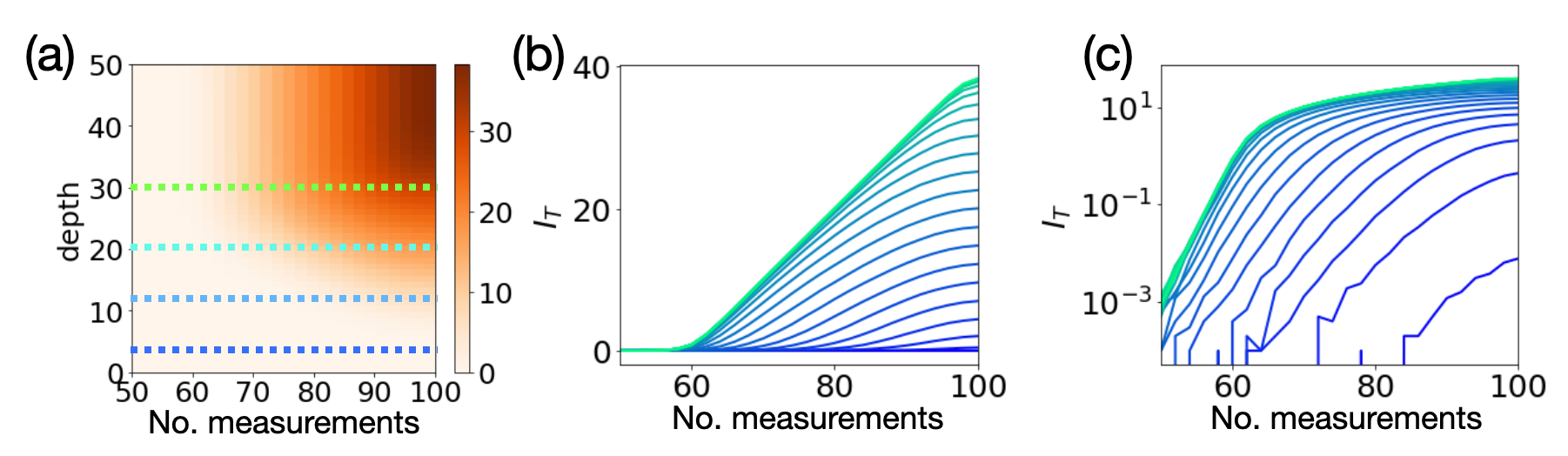}}
\caption{\label{fig:thermalization}(a) We take cross-sections at increasing circuit depth to observe thermalization behavior. (b) Growth of CMI at the cross-setion. (c) Same as (b) but in log scale to observe Page corrections.}
\end{figure}

\subsection{Statistical Mechanics Model for the interpolation}\label{stat_mech}
Finally, we discuss the interpolation from the finite-time teleportation to the Page behavior via the statistical mechanics model. We will try to under the model at the mean-field level which is sufficient to reproduce the numerical observations. We consider a d-dimensional shallow circuit. There is a known procedure \cite{10.1103/physrevb.101.104301} to map the circuit to a (d+1)-dimensional Ising model with ferromagnetic interactions. Since the circuit is shallow, we can integrate out the time direction and reduce to a d-dimension Ising model with coupling scaled with the circuit depth. The entropy can be calculated from the domain-wall free energy at different boundary field conditions. The relevant entropies and the corresponding boundary field conditions are presented in Fig. \ref{fig:stat_mech}(a,b). The spins subject to field in B are those unmeasured in the circuit picture, and they appear randomly. Therefore, the problem reduces to penetrating the magnetic ordering in A and C into B subject to random pinning field in the opposite direction.

We solve the model by using a mean-field ansatz where we assume the entire A, B, and C individually has a mean-field variable \(S_A, S_B, S_C\). To that end, we can write the partition function as
\begin{equation}
    Z_{MF}=\exp(h_A S_A + h_B S_B + h_C S_C -J S_A S_B - J S_B S_C)
\end{equation}
Where \(J\) is the coupling that scales with the circuit depth, and \(h_A, h_B, h_C\) are the summed local field. Since the pinning field in B is randomly placed, we can simply sum them up in the coarse-grained picture to obtain \(h_B\). We compute the individual entropies and \(\delta I\) from this mean-field model and present the data in Fig. \ref{fig:stat_mech}(c,d,f). We observe a very good agreement with the circuit data on the qualitative level.

With the mean-field model, we can explain some of the features as shown in Fig. \ref{fig:stat_mech}(f) which are also observed in the circuit data. At the blue star, there is no pinning field in B, so the problem becomes penetrating magnetic ordering from A and C into a thermally fluctuating but free magnet B. As the coupling increases, we observe a second-order ferromagnetic phase transition as previously investigated by \cite{10.48550/arXiv.2110.06963}. The critical point is \(O(\log  n)\) in 1D and \(O(1)\) in 2D, and we indeed have observed such scaling in the simulation. Once we turn on the pinning potential in B to approach the purple dashed line, the problem becomes penetrating magnetic ordering into a magnet pulled in the reverse direction. Therefore, this becomes a first-order hysteretic phase transition. The overall teleporting phase diagram resembles a magnetic phase diagram shown in Fig. \ref{fig:stat_mech}(e), where the temperature and polarization are the control parameters.
\begin{figure}
\scalebox{0.5}{\includegraphics{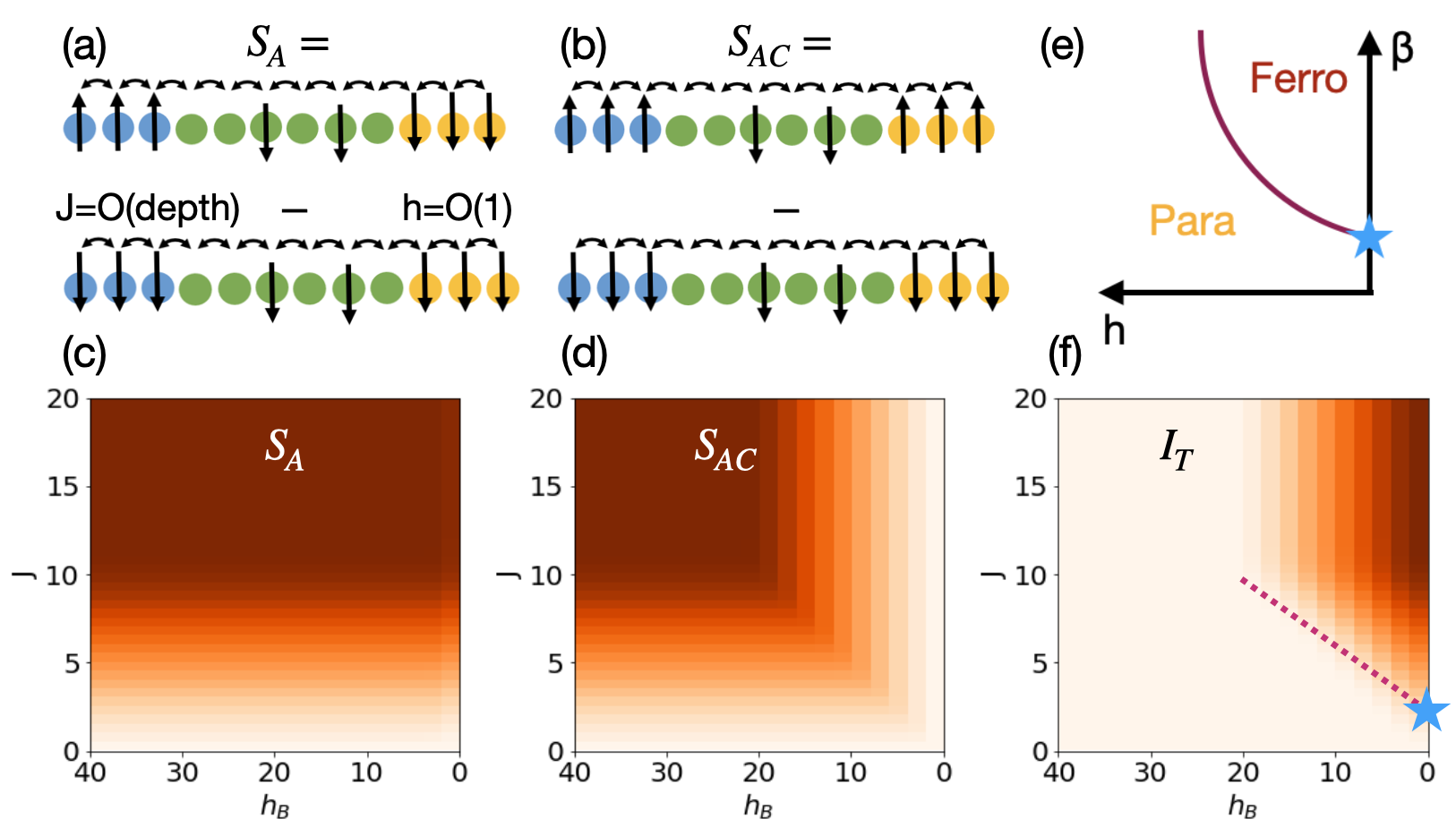}}
\caption{\label{fig:stat_mech}(a) Domain wall energy corresponding to \(S(A)\). (b) Domain wall energy corresponding to \(S(AC)\). (c,d,f) Calculated \(S(A)\), \(S(A)C\), and \(\delta I\) in the mean-field approximation. (blue star) second-order ferromagnetic transition. (red line) first-order hysteretic transition. (e) Analog to the ferromagnetic-paramagnetic phase diagram.}
\end{figure}
\subsection{Scaling of Finite Size Teleportation}\label{scaling}
We finally present additional numerics in this subsection to understand the scaling of the critical point and to support our argument. We calculate \(\delta I\) at the line crossing the critical point as shown in Fig. \ref{fig:finite_size_scaling}(a,d) for the case of measurement channels and erasure channels, respectively. We plot \(\delta I\) as a function of circuit depth at increasing system size in Fig. \ref{fig:finite_size_scaling}(b,e), from where we interpolate back to  \(\delta I=0\) to extract the critical point. We plot the scaling of critical point at increasing system size in Fig. \ref{fig:finite_size_scaling}(c,f). In the case of measurement channels, we observe the \(O(\log(n))\) scaling as expected from the 1D Ising model. In the case of erasure channels, we fit a power law to extract the scaling to be \(O(n^{0.41})\) which is also close to prediction of the Imry-Ma argument.

\begin{figure}
\scalebox{0.5}{\includegraphics{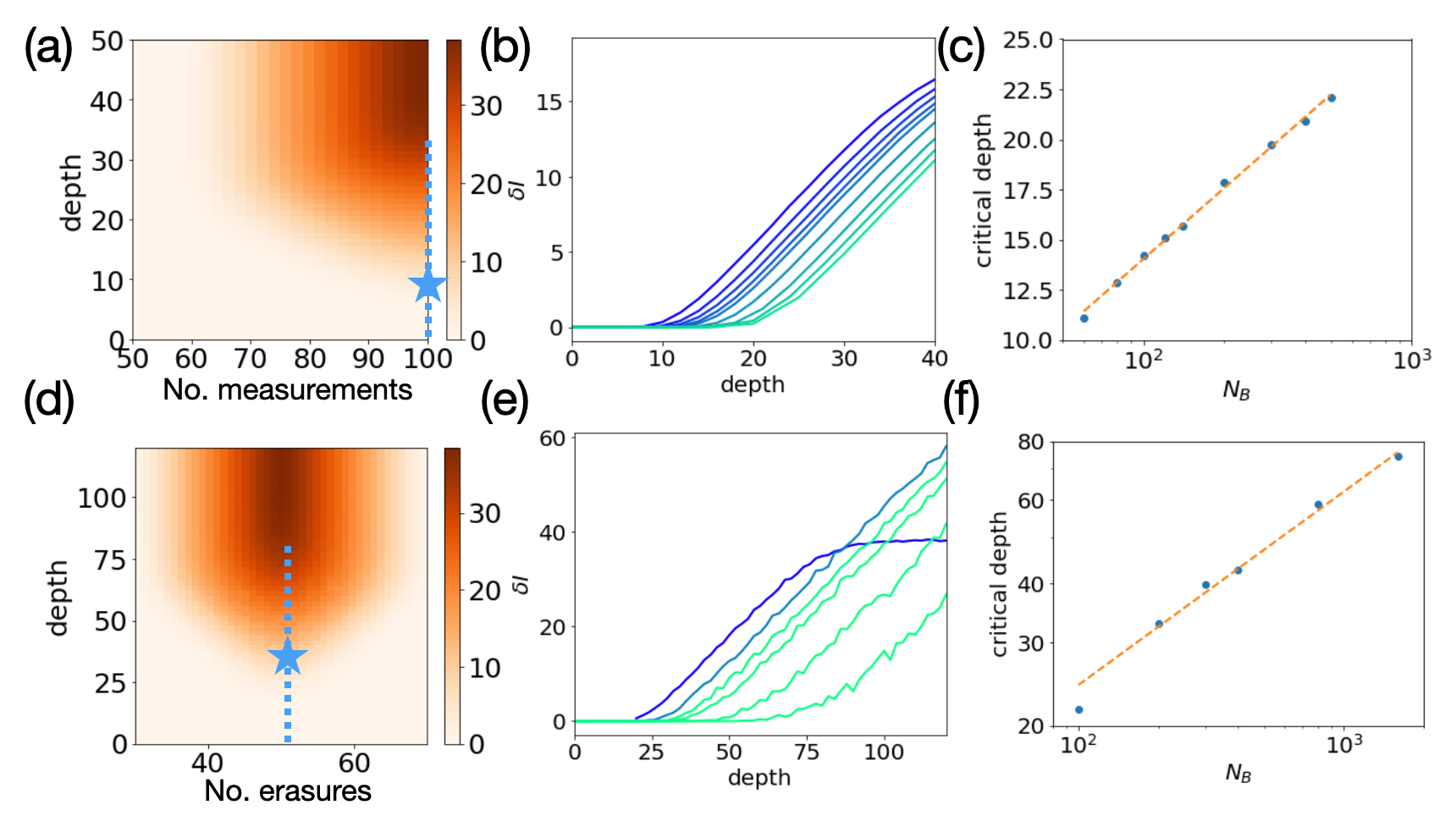}}
\caption{\label{fig:finite_size_scaling}(a) In the case of measurement channels, we probe the scaling numerically when the entire B is measured (dashed blue line), corresponding to the finite-time teleportation. (b) \(\delta I\) at the dashed blue line in (a) for \(N_B\) ranging from 60 to 400. (c) The extrapolated critical point (blue dots) and the logarithmic fit. (d) In the case of erasure channels, we probe the scaling numerically when half of B is measured (dashed blue line). (e) \(\delta I\) at the dashed blue line in (d) for \(N_B\) ranging from 100 to 1600. (c) The extrapolated critical point (blue dots) and the power-law fit with an exponent of 0.41 which is close to the expected value of 0.5.}
\end{figure}

\end{widetext}
%%%%%%%%%%%%%%%%%%%%%%%%% old stuff %%%%%%%%%%%%%%%%%%%%%%%%%

% \nocite{*}

\bibliography{apssamp}% Produces the bibliography via BibTeX.

\end{document}